\documentclass[10pt,journal,compsoc]{IEEEtran}

\newif\ifjournal

\journalfalse

\usepackage{amsmath,amssymb,amsthm}
\usepackage{mathtools}
\usepackage{scalerel}
\usepackage{xspace}
\usepackage{diagbox}
\usepackage{tikz}
\usepackage{circuitikz}
\usetikzlibrary{automata,positioning}
\usetikzlibrary{circuits.logic.US}
\usetikzlibrary{circuits,calc}
\usetikzlibrary{shapes.geometric}
\usetikzlibrary{decorations.pathreplacing}
\usetikzlibrary{backgrounds}

\tikzstyle{input} = [fill=black, isosceles triangle, rotate=90, scale=.3, draw]
\tikzstyle{buffer} = [fill=white, regular polygon, regular polygon sides=3, scale=.6, draw, thick]
\tikzstyle{operator} = [rectangle, draw, rounded corners=2pt, rotate=90, thick]
\tikzstyle{outcirc} = [rectangle, draw, rounded corners=2pt, thick]

\usepackage[keeplastbox]{flushend}

\usepackage{comment}
\usepackage{paralist}
\usepackage{multirow}
\usepackage{adjustbox}
\usepackage{epstopdf}
\usepackage{subcaption}

\usepackage{listings}
\usepackage{color}
\usepackage{soul}
\usepackage{array}

\newtheorem{theorem}{Theorem}[section]
\newtheorem{corollary}[theorem]{Corollary}
\newtheorem{definition}[theorem]{Definition}

\newtheorem{lemma}[theorem]{Lemma}
\newtheorem{observation}[theorem]{Observation}

\newcommand{\NN}{\mathbb{N}}
\newcommand{\BO}{\mathcal{O}}
\newcommand{\BB}[1]{\mathbb{B}^{#1}}
\newcommand{\BBM}[1]{\mathbb{B}_\metas^{#1}}

\newcommand{\outt}{\mathrm{out}}

\newcommand{\rg}{\operatorname{rg}}

\newcommand{\ANDD}{\operatorname{AND}}

\newcommand{\NAND}{\operatorname{NAND}}
\newcommand{\NOR}{\operatorname{NOR}}

\newcommand{\ORR}{\operatorname{OR}}

\newcommand{\rgmax}{\operatorname{max^{\rg}}}
\newcommand{\rgmin}{\operatorname{min^{\rg}}}
\newcommand{\rgmaxM}{\operatorname{max^{\rg}_{\metas}}}
\newcommand{\rgminM}{\operatorname{min^{\rg}_{\metas}}}
\newcommand{\sel}{\operatorname{sel}}

\newcommand{\xmux}{\textsc{xmux}}

\newcommand{\validrg}[1]{\mathcal{S}^{#1}_{\rg}}

\newcommand{\metas}{\textnormal{\texttt{M}}}

\newcommand{\twosort}{\operatorname{2-sort}}

\newcommand{\foursort}{\operatorname{4-sort}}
\newcommand{\sevensort}{\operatorname{7-sort}}

\newcommand{\tensortc}{\operatorname{10-sort_\#}}
\newcommand{\tensortd}{\operatorname{10-sort_d}}

\newcommand{\ppc}{\operatorname{PPC}}

\newcommand{\parity}{\operatorname{par}}
\newcommand{\res}{\operatorname{res}}

\newcommand{\binary}{\operatorname{Bin-comp}}

\usepackage{amsmath}
\usepackage{relsize}
\newcommand{\bigstarr}{\mathop{\raisebox{-.7pt}{\ensuremath{\mathlarger{\mathlarger{\mathlarger{*}}}}}}}
\DeclareMathOperator*{\bigdiamond}{\scalerel*{\diamond}{\textstyle\sum}}
\DeclareMathOperator*{\bigdMop}{~\bigdiamond{\!}\raisebox{-4pt}{\scriptsize
\metas}}
\newcommand{\bigdM}{\!\!\!\bigdMop}

\newcommand{\OP}{\oplus}

\newcommand{\inp}{d}
\newcommand{\outp}{\pi}

\newcommand{\repg}[1]{\langle #1 \rangle}

\def\dash---{\kern.16667em---\penalty\exhyphenpenalty\hskip.16667em\relax}

\IEEEoverridecommandlockouts

\begin{document}

\title{Optimal Metastability-Containing Sorting\\
via Parallel Prefix Computation$^*$\thanks{$^*$This
article generalizes and extends work presented at DATE 2018~\cite{date18}.}}

\author{Johannes~Bund, Christoph~Lenzen, Moti Medina%
\IEEEcompsocitemizethanks{\IEEEcompsocthanksitem Johannes Bund and Christoph
Lenzen are with the Max Planck Institute for Informatics, Saarland
Informatics Campus, 66123 Saarbr\"ucken, Germany. Email:~\texttt{\{jbund,clenzen\}@mpi-inf.mpg.de}%
\IEEEcompsocthanksitem Moti Medina is with the School of
Electrical \& Computer Engineering, Ben-Gurion~University~of~the~Negev, 8410501 Beer
Sheva, Israel. Email:~\texttt{medinamo@bgu.ac.il}}%
}

\IEEEtitleabstractindextext{%
\begin{abstract}
Friedrichs et al.\ (TC 2018) showed that metastability can be \emph{contained}
when sorting inputs arising from time-to-digital converters, i.e.,
measurement values can be correctly sorted \emph{without} resolving
metastability using synchronizers first. However, this work left open whether
this can be done by small circuits. We show that this is indeed possible, by
providing a circuit that sorts Gray code inputs (possibly containing a
metastable bit) and has asymptotically optimal depth and size.

Our solution utilizes the parallel prefix computation (PPC) framework (JACM
1980). We improve this construction by bounding its fan-out by an arbitrary
$f\geq 3$, without affecting depth and increasing circuit size by a small
constant factor only. Thus, we obtain the first PPC circuits with asymptotically
optimal size, constant fan-out, and optimal depth.

To show that applying the PPC framework to the sorting task is
feasible, we prove that the latter can, despite potential metastability, be
decomposed such that the core operation is associative. We obtain asymptotically
optimal metastability-containing sorting networks. We complement these results with
simulations, independently verifying the correctness as well as small size and
delay of our circuits.
\end{abstract}
%
}

\maketitle

\section{Introduction}\label{sec:intro}

Metastability is a fundamental obstacle when crossing clock domains, potentially
resulting in soft errors with critical consequences~\cite{ginosar11tutorial}. As
it has been shown that metastability cannot be avoided
deterministically~\cite{Mar81}, synchronizers~\cite{kinniment08} are employed to
reduce the error probability to tolerable levels. This approach trades precious
time for reliability: the more time is allocated for metastability resolution,
the smaller the probability of metastability-induced faults.

Recently, a different approach has been proposed, coined
\emph{metastability-containing} (MC) circuits~\cite{friedrichs18containing}. It
accepts a limited amount of metastability in the input to a digital circuit and
ensures limited metastability of its output, so that the result\\ is still
useful. In a series of works~\cite{async16,date17,date18}, we applied this
approach to a fundamental primitive: sorting. The circuit given in~\cite{date18}
is asymptotically optimal in depth and size.

\paragraph*{\textbf{Our Contribution}}
In this article, we present the machinery used to obtain the circuit
from~\cite{date18} in detail. We prove that CMOS implementations of basic gates
realize Kleene logic~(cf.~\cite[\S64]{kleene52meta}), justifying the
computational model introduced in~\cite{friedrichs18containing} and used in this
article.

The task of sorting an arbitrary number of inputs can be reduced to sorting two
inputs by using sorting networks~\cite{knuth1998art}. The $0$-$1$-principle
(cf.\ Section~\ref{sec:related}) shows that plugging an MC $\twosort(B)$
circuit (for $B$-bit inputs) into a sorting network (for $n$ values) readily
yields an MC circuit that is capable of sorting $n$ inputs. Hence, we need to
design a $\twosort(B)$ circuit sorting two inputs in an MC way.

As the choice of the encoding matters a lot for MC circuits, we characterize the
set of input strings we want to sort (``valid strings''). A valid string is
either a (standard) Gray code string or a string obtained from a Gray code
string by replacing the unique bit that would change on the up-count to the
``next'' codeword by $\metas$ for metastability (the third logic value in Kleene
logic). When using non-redundant codes, the use of Gray codes is mandatory: when
converting an analog value to a digital one, continuously changing the input can
force any circuit (that uses the value in a non-trivial way) into
metastability~\cite{Mar81}. Moreover, for combinational circuits in the
abstraction of Kleene logic, \emph{all} output bits that change when flipping a
given input bit must become unstable when the input bit is unstable,
cf.~\cite{friedrichs18containing}. For instance, encoding a value unknown to be
$11$ or $12$ in standard binary code would result in a string that, once
metastability has been resolved, may represent any number in the interval from
$8$ to $15$, cf.~Section~\ref{sec:basics}.

Valid strings arise naturally when stopping a Gray code counter
asynchronously~\cite{tdc16} or, more generally, whenever performing
analog-to-digital conversion; respective circuits may risk multiple metastable
bits to achieve better average-case precision, but for the best worst-case
precision one can stick to guaranteeing valid strings as output. Exploiting the
structure of Gray code and the restriction to valid strings, we show how to
reliably sort all inputs despite the uncertainty about the represented value
arising from metastability.

We formally specify the $\twosort(B)$ circuit and then prove that the task of
comparing two valid strings can be decomposed into first performing a
four-valued comparison on each prefix pair of the two valid input strings, and
then inferring the corresponding output bits. This reduces the design of
$\twosort(B)$ to a parallel prefix computation (PPC) problem, which for our
purposes can be phrased as follows.
\begin{definition}[$\ppc_{\OP}(B)$]\label{def:ppc}
For associative $\OP\colon D\times D\to D$ and $B\in \NN$, a
$\ppc_{\OP}(B)$ circuit is specified as follows.
\begin{compactitem}
  \item[\textbf{\emph{Input:}}] $d \in D^B$\:,
  \item[\textbf{\emph{Output:}}] $\pi \in D^B$\:,
  \item[\textbf{\emph{Functionality:}}] $\pi_i=\bigoplus_{j=1}^i d_j$ for all
  $i\in [1,B]$.
\end{compactitem}
\end{definition}
Fast PPC circuits that are simultaneously (asymptotically) optimal in depth and
size are known due to a celebrated result by Ladner and
Fischer~\cite{ladner1980parallel}. Going beyond~\cite{date18}, we present the
full range of solutions that can be derived using their framework, which allows
for a trade-off between depth and size of the $\twosort$ circuit. Most
prominently, optimizing for depth reduces the depth of the circuit by a factor
of $2$ compared to~\cite{date18} to optimal $\lceil \log B\rceil$, at the
expense of increasing the size by a factor of up to $2$.

However, relying on the construction from~\cite{ladner1980parallel} as-is
results in a very large fan-out. We present a modification reducing fan-out to
any number $f\geq 3$ without affecting depth, increasing the size by a factor of
only $1+\BO(1/f)$ (plus at most $3B/2$ buffers). In particular, our results
imply that the depth of an MC sorting circuit can match the delay of a
non-containing circuit, while maintaining constant fan-out and a constant-factor
size overhead. Due to the fact that PPC circuits lie at the heart of fast
adders~\cite{swartzlander15arithmetic}, we consider this result of independent interest.

We complement our theoretical findings by simulations confirming the correctness
and small size of the devised circuits. Post-layout area and delay of
the designed circuits compare favorably with a baseline provided by a
straightforward non-containing implementation.

\paragraph*{\textbf{Organization of this Article}}
We discuss related work in Section~\ref{sec:related}. Some preliminaries, the
computational model and its justification, as well as the problem specification
are given in Section~\ref{sec:basics}. Next, in Section~\ref{sec:operators}, we
break the task of designing a $\twosort(B)$ circuit down into comparing prefixes
and subsequently generating the output bits out of the computed comparison
values and the respective pair of input bits. The comparison can be further
decomposed into sequential application of an associative operator, which enables
application of the PPC framework to compute all prefixes efficiently in parallel
with (asymptotically) optimal depth. In order to keep this article
self-contained, we compactly review the PPC framework in Section~\ref{sec:ppc}.
The section then proceeds to showing how to modify the construction for bounded
fan-out and bounding the size of the resulting circuits. In
Section~\ref{sec:simulation}, we implement the base operators by subcircuits and
plug the pieces together to obtain complete circuits. We then simulate them up
to an input width of $B=16$ to independently verify their correctness, and
provide delay and area of the laid out circuits.
We compare to a non-containing version as baseline, demonstrating the controlled
increase in size of the circuit. We conclude the article in
Section~\ref{sec:conclusion}, where we also briefly discuss follow-up work that
generalizes our results, demonstrating that higher-level concepts of this work
like sorting networks and parallel prefix computation are applicable to further
MC circuits.%

\section{Related Work}\label{sec:related}

\paragraph*{\textbf{Sorting Networks}}
Sorting networks (see, e.g., \cite{knuth1998art}) sort $n$ inputs from a totally
ordered universe by feeding them into $n$ parallel wires that are connected by
$\twosort$ elements, i.e., subcircuits sorting two inputs; these can act in
parallel whenever they do not depend on each other's output. A correct sorting
network sorts all possible inputs, i.e., the wires are labeled $1$ to $n$ such
that the $i^{th}$ wire outputs the $i^{th}$ element of the sorted list of
inputs. The \emph{size} of a sorting network is its number of $\twosort$
elements and its \emph{depth} is the maximum number of $\twosort$ elements an
input may pass through until reaching the output.

The $0$-$1$-principle~\cite{knuth1998art} states that a sorting network\dash---
assuming the $\twosort$ circuits are correct\dash---is correct if and only if it
sorts $0$-$1$ inputs correctly. Thus, we obtain sorting networks for inputs that
may suffer from metastability by constructing $\twosort$ circuits (w.r.t.\ a
suitable order on such inputs) and plugging them into existing sorting networks.

Sorting networks have been extensively studied. Tight lower bounds of depth
$\Omega(\log n)$ (trivial) and size $\Omega(n \log n)$ (see, e.g.,
\cite{cormen09introduction}) are known and can be simultaneously asymptotically
matched~\cite{ajtai83}. More practically, for small values of $n$ optimal depth
and/or size networks are
known~\cite{bundala2014optimal,codish2014twenty,knuth1998art}. Accordingly, our
task boils down to finding optimal (or close to optimal) metastability-containing
$\twosort$ circuits. For $B$-bit inputs, our $\twosort$ circuits have depth and
size $\BO(\log B)$ and $\BO(B)$, respectively, which is (trivially) optimal up
to constants; as size and depth of our circuits are close to non-containing
$\twosort$ circuits (cf.\ Table~\ref{tab:sorting}), we conclude that our
approach yields MC sorting networks that are optimal up to small constant
factors in both depth and size.

\paragraph*{\textbf{Prior Work on MC Circuits}}
Recent work~\cite{friedrichs18containing} shows that for any Boolean function a
combinational MC circuit implementing its \emph{metastable closure} (see
Definition~\ref{def:closure}) exists. The metastable closure can be seen as a
best effort to contain metastability: when for an input with (some) metastable
bits the stable input bits already determine a given output bit of the original
Boolean function, the closure attains the respective value on this output bit;
otherwise it is metastable.

Unfortunately, the proof from~\cite{friedrichs18containing}, which uses a
construction dating back to Huffman~\cite{huffman57design}, yields circuits of
exponential size in the number of input bits $B$. The same is true for
speculative computing~\cite{tarawneh12hiding}. Unconditional lower bounds on MC
circuits~\cite{ikenmeyer18complexity} show that this cannot be avoided in
general, even if the implemented function admits a small non-containing circuit.
The same work provides, assuming that at most $k$ input bits can be metastable,
a construction with multiplicative $B^{\BO(k)}$ and additive $\BO(k\log B)$
overheads in size and depth, respectively. For the $\twosort$ element, $k=2$
(each Gray code string may contain one metastable bit), but the resulting
circuits are still far from optimal.

In~\cite{friedrichs18containing}, an alternative construction relying on
non-combinational logic is given, achieving (up to minor-order terms) factor
$2k+1$ increase in size and additive $\Theta(\log k)$ increase in depth of the
resulting circuit; for a $\twosort$ circuit, $k=2$, so these overheads are
constant. Rule-of-thumb calculations suggest that optimized versions of the
circuits presented here and derived by this method would have comparable
performance. A fair and detailed comparison would require fully-fledged designs
of both approaches, which is beyond the scope of this article. Note, however,
that our design has the advantage of being purely combinational.

\paragraph*{\textbf{Parallel Prefix Computation}}
Ladner and Fischer \cite{ladner1980parallel} studied the parallel application of
an associative operator to all prefixes of an input string of length $\ell$
(over an arbitrary alphabet). They give parallel prefix computation (PPC)
circuits of depth $\BO(\log \ell)$ and size $\BO(\ell)$ (where the circuit
implementing the operator is assumed to have size and depth $1$). However, when
requiring optimal depth of $\lceil \log \ell\rceil$, their corresponding
solution suffers from fan-out larger than $\ell/2$. An earlier construction by
Kogge and Stone~\cite{kogge73recurrence} simultaneously achieves optimal depth
and fan-out of $2$. This yields the fastest adder circuits to
date (cf.~\cite{swartzlander15arithmetic}), but at the expense of a large size
of $\ell (\lceil\log \ell \rceil - 1)+1$. A number of additional constructions
have been developed for adders, including special
cases (\cite{brent82adders,sklansky60addition}) of the one by Ladner and
Fischer, cf.~\cite{zimmermann97adder}. However, no other construction achieves
asymptotically optimal depth and size.

\section{Model and Problem}\label{sec:basics}

In this section, we discuss how to model metastability in a worst-case fashion
and formally specify the input/output behavior of our circuits. Our model is a
simplified version of the one from~\cite{friedrichs18containing} for
combinational circuits (cf.~\cite[Chap.~7]{friedrichs17diss}). This means to
represent metastable ``bits'' by $\metas$ and extend truth tables as in Kleene's
$3$-valued logic~\cite[\S64]{kleene52meta}.

\paragraph*{\textbf{Basic Notation}}
We set $[N]:= \{0,\ldots,N-1\}$ for $N\in \mathbb{N}$ and $[i,j] = \{i, i+1,
\ldots, j\}$ for $i,j\in \NN$, $i\leq j$. We denote $\BB{}\coloneqq\{0,1\}$ and
$\BBM{}\coloneqq\{0,1,\metas\}$. For a $B$-bit string $g \in
\BBM{B}$ and $i\in[1,B]$, denote by $g_i$ its $i$-th bit, i.e., $g=g_1g_2\ldots
g_B$. We use the shorthand $g_{i,j}:=g_i\ldots g_j$, where $i,j\in[1,B]$ and
$i\leq j$. Let $\parity(g)$ denote the parity of $g\in \BB{B}$, i.e, $\parity(g)
= \sum_{i=1}^{B}g_i\bmod 2$. For a function $f$ and a set $A$ we abbreviate
$f(A):=\{f(y)\,|\,y\in A\}$.

\subsection{Binary Reflected Gray Code}


A standard binary representation of inputs is unsuitable: uncertainty of the
input values may be arbitrarily amplified by the encoding. E.g.\ representing a
value unknown to be $11$ or $12$, which are encoded as $1011$ resp.\ $1100$,
would result in the bit string $1\metas\metas\metas$, i.e., a string that is
metastable in every position that differs for both strings. However,
$1\metas\metas\metas$ may represent any number in the interval from $8$ to $15$,
amplifying the initial uncertainty of being in the interval from $11$ to $12$.
An encoding that does not lose precision for consecutive values is Gray code.

We use $B$-bit binary reflected Gray code, $rg_B:[N]\to\BB{B}$, which is defined
recursively. For simplicity (and without loss of generality) we set $N := 2^B$.
A $1$-bit code is given by $\rg_1(0)=0$ and $\rg_1(1)=1$. For $B>1$, we start
with the first bit fixed to $0$ and counting with $\rg_{B-1}(\cdot)$ (for the
first $2^{B-1}$ codewords), then toggle the first bit to $1$, and finally
``count down'' $\rg_{B-1}(\cdot)$ while fixing the first bit again,
cf.~Table~\ref{table:graystruct}. Formally, this yields for $x\in[N]$
\begin{equation*}
\rg_B(x):=\begin{cases}
0\rg_{B-1}(x)& \mbox{if }x\in [2^{B-1}]\\
1\rg_{B-1}(2^B-1-x)& \mbox{if }x\in [2^B]\setminus [2^{B-1}]\:.
\end{cases}
\end{equation*}
As each $B$-bit string is a codeword, the code is a
bijection and the encoding function also defines the decoding function. Denote by
$\repg{\cdot}:\BB{B}\to[N]$ the decoding function of a Gray code string, i.e., for $x\in
[N]$, $\repg{\rg_B(x)}=x$.

\begin{table}
\begin{center}
\caption{4-bit binary reflected Gray code}
\label{table:graystruct}
\begin{tabular}{| c   c| c   c  || c   c  | c   c |}
\hline
 $\#$ & $g_1,g_{2,4}$ & $\#$ & $g_1,g_{2,4}$ & $\#$ & $g_1,g_{2,4}$ & $\#$ & $g_1,g_{2,4}$\tabularnewline
\hline
\hline
$0$           & $0\;000$ &$4$           & $0\;110$ & $8$           & $1\;100$ & $12$           & $1\;010$\tabularnewline
$1$           & $0\;001$ & $5$           & $0\;111$ & $9$           & $1\;101$ & $13$           & $1\;011$ \tabularnewline
$2$           & $0\;011$ & $6$           & $0\;101$ & $10$           & $1\;111$ & $14$           & $1\;001$ \tabularnewline
$3$           & $0\;010$ & $7$           & $0\;100$ & $11$           & $1\;110$ & $15$           & $1\;000$ \tabularnewline
\hline
\end{tabular}
\end{center}
\end{table}

For two binary reflected Gray code strings $g,h\in \BB{B}$, we define their
maximum and minimum as
\begin{align*}
\left(\rgmax\{g,h\},\rgmin\{g,h\}\right)&:=\begin{cases}
(g,h) & \mbox{if }\repg{g}\geq \repg{h}\\
(h,g) & \mbox{if }\repg{g}< \repg{h}\:.
\end{cases}
\end{align*}
For example:
\begin{itemize}
  \item$\rgmax\{0011,0100\}=\rgmax\{\rg_B(2),\rg_B(7)\}=0100$,
  \item$\rgmin\{0111,0101\}=\rgmin\{\rg_B(9),\rg_B(10)\}=0111$.
\end{itemize}

\subsection{Valid Strings}

The inputs to the sorting circuit may have some metastable bits, which means that
the respective signals behave out-of-spec from the perspective of Boolean logic.
Such inputs, referred to as \emph{valid strings}, are introduced with the help
of the following operator.
\begin{definition}[$*$ Operator]\label{def:star}
For $B\in \NN$, define the operator
$*: \BBM{B} \times \BBM{B} \rightarrow \BBM{B}$ by
\begin{equation*}
\forall i\in \{1,\ldots,B\}:(x * y)_i := \begin{cases}
x_i & \mbox{if }x_i=y_i\\
\metas & \mbox{else.}
\end{cases}
\end{equation*}
\end{definition}
\begin{observation}\label{obs:bigstarr}
The operator $*$ is associative and commutative. Hence, for a set
$S=\{x^{(1)},\ldots,x^{(k)}\}$ of $B$-bit strings, we can use the shorthand
$\bigstarr S := \bigstarr_{x\in S} x := x^{(1)}*x^{(2)}*\ldots*x^{(k)}.$
We call $\bigstarr S$ \emph{the superposition of the strings in $S$.}
\end{observation}

Valid strings have at most one metastable bit. If this bit resolves to either
$0$ or $1$, the resulting string encodes either $x$ or $x+1$ for some $x$,
cf.~Table~\ref{table:validinputs}.
\begin{definition}[Valid Strings]\label{def:validinput}
Let $B\in \mathbb{N}$ and $N=2^B$. Then, the set of \emph{valid strings of
length $B$} is
\begin{equation*}
\validrg{B}:=\rg_B([N])\cup \bigcup_{x\in [N-1]}
\{\rg_B(x)*\rg_B(x+1)\}\:.
\end{equation*}
\end{definition}

\begin{table}
\begin{center}
\caption{$4$-bit valid inputs}
\label{table:validinputs}
\begin{tabular}{| c | c || c  | c || c  |  c  || c | c |}
\hline
 $g$ & $\repg{g}$ & $g$ & $\repg{g}$ & $g$ & $\repg{g}$ & $g$ & $\repg{g}$\tabularnewline
\hline
\hline
$0000$&$0$   & $0110$ &$4$   & $1100$ &$8$   & $1010$ &$12$\tabularnewline
$000\metas$&$-$ & $011\metas$ &$-$ & $110\metas$ &$-$ &$101\metas$&$-$\tabularnewline
$0001$&$1$   & $0111$ &$5$   & $1101$ &$9$   & $1011$ &$13$\tabularnewline
$00\metas1$&$-$ & $01\metas1$ &$-$ & $11\metas1$ &$-$ &$10\metas1$&$-$\tabularnewline
$0011$&$2$   & $0101$ &$6$   & $1111$ &$10$  & $1001$ &$14$\tabularnewline
$001\metas$&$-$ & $010\metas$ &$-$ & $111\metas$&$-$&$100\metas$&$-$\tabularnewline
$0010$&$3$   & $0100$ &$7$   & $1110$ &$11$  & $1000$ &$15$\tabularnewline
$0\metas10$&$-$ & $\metas100$ &$-$ & $1\metas10$ &$-$& $-$ &$-$\tabularnewline
\hline
\end{tabular}
\end{center}
\end{table}

\subsection{Resolution and Closure}

To extend the specification of $\rgmax$ and $\rgmin$ to valid strings, we
make use of the \emph{metastable closure}~\cite{friedrichs18containing}. The
metastable closure is defined over the possible \emph{resolutions} of
metastable bits.
\begin{definition}[Resolution~\cite{friedrichs18containing}]\label{def:resolution}
For $x\in \BBM{B}$, define the resolution $\res(x):\BBM{B}\rightarrow
\mathcal{P}\left(\BB{B}\right)$ as follows:
\begin{equation*}
\res(x):=\{y\in \BB{B}\,|\,\forall i\in \{1,\ldots,B\}\colon x_i\neq
\metas\Rightarrow y_i=x_i\}\:.
\end{equation*}

\end{definition}
Thus, $\res(x)$ is the set of all strings obtained by replacing all $\metas$s in
$x$ by either $0$ or $1$: $\metas$ acts as a ``wild card.''
For any $x$ and $y$, we have that $\res(xy)=\res(x)\res(y)$.

We note two observations for later use.
\begin{observation}\label{obs:identity} For any $x\in\BBM{B}$, $\bigstarr
\res(x)=x$.
\end{observation}

\begin{proof}
Let $x\in\BBM{B}$ and let $I$ be the set of indices where $x$ is stable, i.e.,
$i\in I$ iff $x_i\neq\metas$. From Definition~\ref{def:resolution}, we get that
\begin{equation*}
\forall i\in\{1\ldots B\}:i\in I \Leftrightarrow \{x_i\}=\res(x_i)\:.
\end{equation*}
By Definition~\ref{def:star} and Observation~\ref{obs:bigstarr},
\begin{equation*}
\forall i\in \{1,\ldots,B\}:\{(\bigstarr \res(x))_i\} = \begin{cases}
\res(x_i) & \mbox{if }i\in I\\
\{\metas\} & \mbox{else.}
\end{cases}
\end{equation*}
This entails the claim $\bigstarr \res(x)=x$.
\end{proof}

For example: $\bigstarr \res(0\metas10) = \bigstarr \{0010,0110\}
=0\metas10\,.$
\begin{observation}\label{obs:subset}
For $\emptyset\neq S\subseteq\BB{B}$, we have $S\subseteq \res(\bigstarr S)$.
\end{observation}

\begin{proof}
Let $\emptyset \neq S\subseteq\BB{B}$ and $s\in S$. Define $I$ as
the set of indices where $\bigstarr S$ is stable, i.e.,
$$\forall i\in\{1\ldots B\}:i\in I \Leftrightarrow (\bigstarr S)_i\neq\metas.$$
From Definition~\ref{def:star} and Observation~\ref{obs:bigstarr}, we conclude
for $i\in\{1\ldots B\}$:
$$(\bigstarr S)_i=\begin{cases}s_i&\mbox{if }i\in I\\
\metas&\mbox{else.}\end{cases}$$
Since by Definition~\ref{def:resolution} each combination of replacing $\metas$s
by $0$s and $1$s occurs in $\res(\bigstarr S)$, we conclude that
$$\exists x \in\res(\bigstarr S):\forall i\notin I:x_i=s_i.$$
Since $\forall i\in\{1\ldots B\}:i\in I \Leftrightarrow \{x_i\}=\res(x_i)$, $s\in\res(\bigstarr S)$.
This proves the claim $S\subseteq\res(\bigstarr S)$.
\end{proof}

We observe that in general the reverse direction does not hold, i.e., $\res(\bigstarr S) \nsubseteq S$.
For example, consider $S = \{01,10\}$ and thus $\bigstarr S = \metas\metas$ such
that $\res(\bigstarr S) = \{00,01,10,11\} = \BB{2}$. Hence, $S\subseteq \res(\bigstarr S)$ but not
$\res(\bigstarr S)\subseteq S$. In contrast, for $|\res(\bigstarr S)| \leq
2$, we can see that the reverse direction holds.
\begin{observation}\label{obs:reverseresolution}
For any subset of strings $S\subseteq \BB{B}$, if $|\res(\bigstarr S)| \leq
2$, then $\res(\bigstarr S) = S$.
\end{observation}

\begin{proof}
Since $\bigstarr S$ can contain at most one $\metas$ bit, we know that $S$ can contain at most
two strings that differ in one position. It is then straightforward to show that every string in
$\res(\bigstarr S)$ is element of $S$. Together with Observation~\ref{obs:subset} this shows the equality.
\end{proof}

The metastable closure of an operator on binary inputs extends it to inputs that
may contain metastable bits. This is done by considering all resolutions of the
inputs, applying the operator, and taking the superposition of the results.
\begin{definition}[The $\metas$
Closure~\cite{friedrichs18containing}]\label{def:closure} Given an operator
$f\colon \BB{n} \to \BB{m}$, its \emph{metastable closure
$f_{\metas}\colon \BBM{n} \to \BBM{m}$} is defined by
$f_{\metas}(x):= \bigstarr \{f(x')|x'\in\res(x)\}$. Recalling the basic notation
we abbreviate this by $f_{\metas}(x)= \bigstarr f(\res(x))$.
\end{definition}
The closure is the best one can achieve w.r.t.\ containing
metastability with clocked logic using standard
registers~\cite{friedrichs18containing}, i.e., when $f_{\metas}(x)_i=\metas$, no
such implementation can guarantee that the $i^{th}$ output bit stabilizes in a
timely fashion.

\subsection{Output Specification}
We want to construct a circuit computing the maximum and minimum of two valid
strings, enabling us to build sorting networks for valid strings.
First, however, we need to answer the question what it means to ask for the
maximum or minimum of valid strings. To this end, suppose a valid string is
$\rg_B(x)*\rg_B(x+1)$ for some $x\in [N-1]$, i.e., the string contains a
metastable bit that makes it uncertain whether the represented value is $x$ or
$x+1$.
If we wait for metastability to resolve, the
string will stabilize to either $\rg_B(x)$ or $\rg_B(x+1)$. Accordingly, it
makes sense to consider $\rg_B(x)*\rg_B(x+1)$ ``in between'' $\rg_B(x)$ and
$\rg_B(x+1)$, resulting in the following total order on valid strings
(cf.~Table~\ref{table:validinputs}).
\begin{definition}[$\prec$]\label{def:order}
We define a total order $\prec$ on valid strings as follows. For $g,h\in
\BB{B}$, $g\prec h\Leftrightarrow \repg{g}<\repg{h}$. For each $x\in [N-1]$, we
define $\rg_B(x)\prec \rg_B(x)*\rg_B(x+1)\prec \rg_B(x+1)$. We extend the
resulting relation on $\validrg{B}\times\validrg{B}$ to a total order by taking
the transitive closure. Note that this also defines $\preceq$, via $g\preceq
h\Leftrightarrow (g=h\vee g\prec h)$.
\end{definition}
We intend to sort with respect to this order. It turns out that implementing a
$\twosort$ circuit w.r.t.\ this order amounts to implementing the metastable
closure of $\rgmax$ and $\rgmin$.

\begin{lemma}\label{lem:twosort_closure}
Let $g,h\in \validrg{B}$. Then
\begin{equation*}
g\preceq h\Leftrightarrow (\rgmaxM\{g,h\},\rgminM\{g,h\})=(h,g)\:.
\end{equation*}
\end{lemma}

\begin{proof}
If $g\prec h$, Definitions~\ref{def:validinput} and~\ref{def:order} imply for
all $g'\in \res(g)$ and all $h'\in \res(h)$ that $g'\preceq h'$
(cf.~Table~\ref{table:validinputs}).
Observation~\ref{obs:identity} shows that
$\bigstarr \res(g)=g$ for any $g\in \validrg{B}$. From
Definition~\ref{def:closure}, we can thus conclude that
$\rgmaxM\{g,h\}=\bigstarr \res(h)=h$ and $\rgminM\{g,h\}=\bigstarr \res(g)=g$.

If $h\prec g$, analogous reasoning shows that
\begin{equation*}
(\rgmaxM\{g,h\},\rgminM\{g,h\})=(g,h)\neq (h,g)\:.
\end{equation*}
The remaining case is that $g=h$. In the case where $g$ does not contain an
$\metas$ bit, we have $\rgmaxM\{g,h\}=\rgmax\{g,h\}=g=h$. Considering the second
case ($g=h=\rg_B(x)*\rg_B(x+1)$, for $x\in[2^B-1]$), we get that
$\rgmaxM\{g,h\}= \bigstarr\{\rg_B(x),\rg_B(x+1)\}=\rg_B(x)*\rg_B(x+1)=h$.
Likewise, $\rgminM\{g,h\}=h=g$.
\end{proof}

In other words, $\rgmaxM$ and $\rgminM$ are the $\max$ and $\min$ operators
w.r.t.\ the total order on valid strings shown in Table~\ref{table:validinputs},
e.g.,
\begin{compactitem}
  \item $\rgmaxM\{1001,1000\}=\rg_4(15)=1000$,
  \item $\rgmaxM\{0\metas10,0010\}=\rg_4(3)*\rg_4(4)=0\metas10$, and
  \item $\rgmaxM\{0\metas10,0110\}=\rg_4(4)=0110$.
\end{compactitem}
Hence, our task is to implement $\rgmaxM$ and $\rgminM$.
\begin{definition}[$\twosort(B)$]\label{def:twosort}
For $B\in \NN$, a $\twosort(B)$ circuit is specified as follows.
\begin{compactitem}
  \item[\textbf{\emph{Input:}}] $g,h \in \validrg{B}$\:,
  \item[\textbf{\emph{Output:}}] $g',h' \in \validrg{B}$\:,
  \item[\textbf{\emph{Functionality:}}] $g'=\rgmaxM\{g,h\}$, $h'=\rgminM\{g,h\}$.
\end{compactitem}
\end{definition}

\subsection{Computational Model and CMOS Logic}
We seek to use standard components and combinational logic only. We use the
model of \cite{friedrichs18containing}, which specifies the behavior of basic gates on
metastable inputs via the metastable closure of their behavior on binary inputs,
cf.~Table~\ref{tab:gates}. We use the standard notational convention that $a+b
= \ORR_{\metas}(a,b)$ and $ab = \ANDD_{\metas}(a,b)$.

\begin{table}
\centering
\caption{Extensions to metastable inputs of AND (left), OR (center), and an
inverter (right) according to Kleene logic.}\label{tab:gates}
\begin{tabular}{c|ccc}
  \diagbox[width=2.4em,height=2.4em]{b}{a} & 0 & 1 & \metas\\ \hline
  0 & 0 & 0 & 0\\
  1 & 0 & 1 & \metas\\
  \metas & 0 & \metas & \metas
\end{tabular}
\qquad
\begin{tabular}{c|ccc}
  \diagbox[width=2.4em,height=2.4em]{b}{a} & 0 & 1 & \metas\\ \hline
  0 & 0 & 1 & \metas\\
  1 & 1 & 1 & 1\\
  \metas & \metas & 1 & \metas
\end{tabular}
\qquad
\begin{tabular}{c|c}
a & $\bar{\text{a}}$\\ \hline
0 & 1\\
1 & 0\\
\metas & \metas
\end{tabular}
\end{table}

Note that in this logic, most familiar identities hold: $\ANDD$ and $\ORR$ are
associative, commutative, and distributive, and DeMorgan's laws hold. However,
naturally the law of the excluded middle becomes void. For instance, in general,
$\ORR(x,\bar{x})\neq 1$, as $\ORR(\metas,\metas) = \metas$.

We now argue that basic CMOS gates behave according to this logic, justifying
the model. For the sake of an intuitive notation, we apply some slightly unusual
conventions. In the following, let $R_1$ be a wildcard that can refer to any
resistance that is ``low'', i.e., close to being negligible, as e.g.\ that of a
transistor in its stable conducting state (i.e., any PMOS transistor subjected
to a low gate voltage or any NMOS transistor subjected to a high
gate voltage). Similar, denote by $R_0$ any resistance that is ``high'', i.e., large
compared to $R_1$, such as the resistance of a transistor in its stable
non-conducting state. Thus, with a stable input $b\in \BB{}$ (where we identify
$0$ with low and $1$ with high voltage), an NMOS transistor attains resistance
$R_b$, while a PMOS transistor attains resistance $R_{\bar{b}}$. We can extend
this to unstable inputs $\metas$ by making the conservative assumption that
$R_{\metas}$ is an arbitrary (possibly time-dependent) resistance.

With this notation, we can see that parallel and serial composition of
transistors implements $\ANDD$ and $\ORR$ in Kleene logic, respectively.
\begin{lemma}\label{lem:nmos}
For $k\in \NN$ sufficiently small so that $kR_1\ll R_0$, let $a_1,\ldots,a_k\in
\BBM{}$ be input signals fed to $k$ NMOS transistors interconnected (i) in
parallel or (ii) sequentially. Set $\sigma\coloneqq \sum_{i=1}^k a_i$ and $\pi
\coloneqq \prod_{i=1}^k a_i$, i.e., the $\ORR$ resp. $\ANDD$ over all inputs.
Then the resistance between input and output
of the resulting subcircuit is (roughly) (i) $R_{\sigma}$ resp.\ (ii) $R_{\pi}$.
\end{lemma}

\begin{proof}
Denote by $R$ the resistance between the input and output of the subcircuit.
Suppose first that $\sigma = 0$, i.e., $a_i=0$ for all $i$. Then, for parallel
composition, we get that
$1/R=\sum_{i=1}^k 1/R_{a_i}=k/R_0$,
yielding that $R\geq R_0/k$. On the other hand, if $\sigma = 1$, there is an
index $i$ such that $a_i=1$, yielding for parallel composition that
$1/R=\sum_{i=1}^k 1/R_{a_i}\geq 1/R_1$
and thus $R\leq R_1$. This shows (i).

Now consider sequential composition and suppose first that $\pi=1$, i.e.,
$a_i=1$ for all $i$. Thus,
$R = \sum_{i=1}^k R_{a_i}\leq k R_1$.
In case $\pi=0$, there is some index $i$ so that $a_i=0$, implying that
$R = \sum_{i=1}^k R_{a_i}\geq R_0$.
\end{proof}

The same arguments apply to PMOS transistors.
\begin{corollary}\label{cor:pmos}
For $k\in \NN$ sufficiently small so that $kR_1\ll R_0$, let $a_1,\ldots,a_k\in
\BBM{}$ be input signals fed to $k$ PMOS transistors interconnected (i) in
parallel or (ii) sequentially. Set $\sigma\coloneqq \sum_{i=1}^k \bar{a}_i$ and
$\pi \coloneqq \prod_{i=1}^k \bar{a}_i$, i.e., the $\ORR$ resp. $\ANDD$ over all inputs.
Then the resistance between input and
output of the resulting subcircuit is (roughly) (i) $R_{\sigma}$ resp.\ (ii)
$R_{\pi}$.
\end{corollary}
We remark that the factor of $k$ reduction in the gap between $R_1$ and $R_0$
may imply that a gate's output signal needs to be regenerated using a buffer.
However, this is the same behavior as for logic that assumes stable signals
only, so standard CMOS design techniques account for this.

\begin{figure}
\centering
\resizebox{.2\columnwidth}{!}{
\begin{circuitikz}\draw
node[pmos, thick] (pmosA) {}
(pmosA.D) node[nmos, anchor=D, thick] (nmosA) {}
(pmosA.G) node[anchor=east] {A}
(nmosA.G) node[anchor=east] {A}
(pmosA.S) node[] (vdd) {}
(vdd) node[right] {$V_{DD}$}
node[right=of nmosA.D] (o) {}
(nmosA.S) node[ground, thick, label={[label distance=8pt]-10:$V_{SS}$}] (gnd) {}
(nmosA.D) node[circ]{} [thick] -- (o)
;
\end{circuitikz}}
\resizebox{.35\columnwidth}{!}{
\begin{circuitikz}\draw
node[pmos, thick] (pmosA) {}
node[pmos, right=of pmosA, thick] (pmosB) {}
(pmosA.D) node[nmos, anchor=north west, thick] (nmosA) {}
(nmosA.S) node[nmos, anchor=D, thick] (nmosB) {}
(pmosA.G) node[anchor=east] {A}
(pmosB.G) node[anchor=east] {B}
(nmosA.G) node[anchor=east] {A}
(nmosB.G) node[anchor=east] {B}
node[left = .9 of pmosB.S] (vddconn) {}
node[above = .3 of vddconn] (vdd) {}
(vdd) node[right] {$V_{DD}$}
node[right=of pmosB.D] (o) {}
(nmosB.S) node[ground, thick, label={[label distance=8pt]-10:$V_{SS}$}] (gnd) {}
(vdd) [thick] -- ++(down:.6)
(pmosA.S) -- (pmosB.S)
(pmosA.D) -- (o)
(vddconn) node[circ] {}
(pmosB.D) node[circ] {}
(nmosA.D) node[circ] {}
;
\end{circuitikz}}
\resizebox{.35\columnwidth}{!}{
\begin{circuitikz}\draw
node[pmos, thick] (pmosA) {}
(pmosA.D) node[pmos, anchor=S, thick] (pmosB) {}
(pmosB.D) node[nmos, anchor=north west, thick] (nmosB) {}
node[nmos, left=of nmosB, thick] (nmosA) {}
(pmosA.G) node[anchor=east] {A}
(pmosB.G) node[anchor=east] {B}
(nmosA.G) node[anchor=east] {A}
(nmosB.G) node[anchor=east] {B}
(pmosA.S) node[above=.3] (vdd) {}
(vdd) node[right] {$V_{DD}$}
node[right=of nmosB.D] (o) {}
node[left = .9 of nmosB.S] (gndconn) {}
node[circ] at (gndconn){}
(gndconn) node[ground, thick, label={[label distance=8pt]-10:$V_{SS}$}] (gnd) {}
(nmosA.S) [thick] -- (nmosB.S)
(nmosA.D) -- (o)
(pmosA.S) -- (vdd)
(pmosB.D) node[circ] {}
(nmosB.D) node[circ] {}
;
\end{circuitikz}}
\caption{Standard transistor-level implementations of inverter (left),
$\NAND$ (center), and $\NOR$ (right) gates in CMOS technology. The
latter can be turned into $\ANDD$ and $\ORR$, respectively, by appending an
inverter.}
\label{fig:transistorgates}
\end{figure}
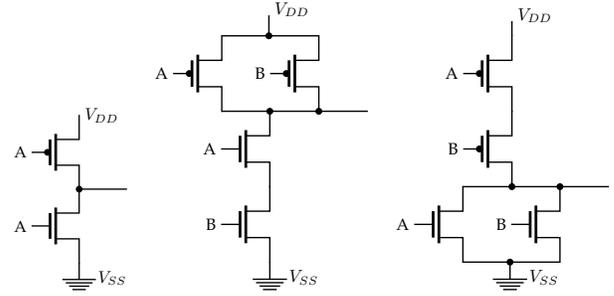

From the above observations, we can readily infer that standard CMOS gate
implementations behave according to Kleene logic in face of potentially
metastable signals, justifying the model from~\cite{friedrichs18containing}.
\begin{theorem}\label{thm:gates_correct}
The CMOS gates depicted in Figure~\ref{fig:transistorgates} implement the
truth tables given in Table~\ref{tab:gates}.
\end{theorem}

\begin{proof}
The output of the gates is $1$ (high voltage) if the resistances from $V_{DD}$
and $V_{SS}$ to the output are low (i.e., roughly $R_1$) and high ($R_0$),
respectively. Similarly, it is $0$ if the roles are reversed. Thus,
Lemma~\ref{lem:nmos} and Corollary~\ref{cor:pmos} show the claim for stable
entries of the truth tables. For the unstable ones, setting $R_{\metas}$ (which
is a wildcard for an arbitrary resistance) to $R_0$ or $R_1$, respectively,
leads to different outcomes. Thus, the output voltage may attain almost any
value between $V_{DD}$ and $V_{SS}$, i.e., the output is $\metas$.
\end{proof}

Similar reasoning applies to many gates, e.g., NAND and NOR gates. We stress,
however, that the property of implementing the closure of the function computed
by the gate on stable values is not universal for CMOS logic. For instance,
standard transistor-level multiplexer implementations do not handle
metastability well, cf.~\cite{friedrichs2017efficient}.

\section{Decomposition of the Task}\label{sec:operators}

In this section, we show that computing $\rgmaxM\{g,h\}$ and $\rgminM\{g,h\}$
for valid strings $g,h\in \validrg{B}$ can be broken down into composing
simple operators in $\BBM{2}\times \BBM{2}\to \BBM{2}$.

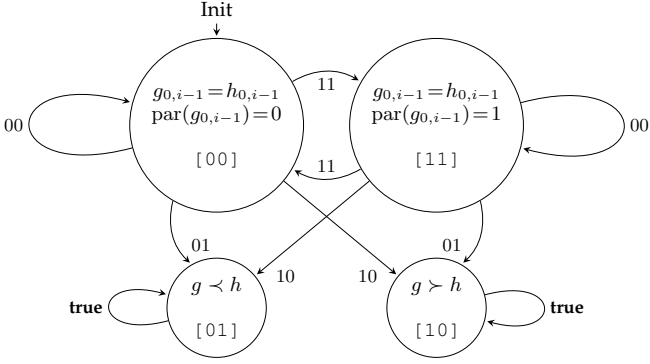
\begin{figure}[t!]
  \centering \small
\resizebox{\columnwidth}{!}{
\begin{tikzpicture}
[>=stealth,shorten >=1pt,auto,node distance=0.7cm, initial above,initial text=Init,scale=0.5]
   \node[state,initial,align=center] (00)   {$g_{0,i-1}\!=\!h_{0,i-1}$\\$\parity(g_{0,i-1})\!=\!0$\\\\\texttt{[00]}};
   \node[state] (11) [right=of 00,align=center] {$g_{0,i-1}\!=\!h_{0,i-1}$\\$\parity(g_{0,i-1})\!=\!1$\\\\\texttt{[11]}};
   \node[state] (01) [below =of 00,align=center] {$g\prec h$\\\\\texttt{[01]}};
   \node[state](10) [below =of 11,align=center] {$g\succ h$\\\\\texttt{[10]}};

   \begin{scope}[every node/.style={scale=.9}]
    \path[->]
    (00) edge [bend left,swap,align=center]node {$11$} (11)
          edge [loop left,align=center] node {$00$} ()
           edge  node [swap,very near end,align=center] {$10$}  (10)
             edge  [bend right,align=center,very near end] node {$01$} (01)
    (01) edge  [loop left] node  {$\textbf{true}$} ()
    (10) edge  [loop right] node  {$\textbf{true}$} ()
    (11) edge  [bend left,align=center] node [swap] {$11$} (00)
          edge  node [near start,align=center,very near end]  {$10$} (01)
           edge  [bend left,align=center,swap, very near end] node {$01$} (10)
           edge [loop right,align=center] node {$00$} ();
    \end{scope}
\end{tikzpicture}}
  \caption{Finite state machine determining which of two Gray code
  inputs $g,h\in \BB{B}$ is larger. In each step, it receives
  $g_ih_i$ as input. State encoding is given in square
  brackets.}\label{fig:fsm}
\end{figure}

\subsection{Comparing Stable Gray Codes via an FSM}\label{sec:stable}

Figure~\ref{fig:fsm} depicts a finite state machine performing a four-valued
comparison of two Gray code strings. In each step of processing inputs $g,h\in
\BB{B}$, it is fed the pair of $i^{th}$ input bits $g_ih_i$. In the following,
we denote by $s^{(i)}(g,h)$ the state of the machine after $i$ steps, where
$s^{(0)}(g,h)\coloneqq 00$ is the starting state. For ease of notation, we will
omit the arguments $g$ and $h$ of $s^{(i)}$ whenever they are clear from context.
Table~\ref{tab:exampleone} shows an example of a run of the finite state machine.
\begin{table}[t]
  \centering
  \caption{Run of the FSM on inputs $g=1001$ and $h=1000$}
  \label{tab:exampleone}
\begin{tabular}{c|c|c|c|c|c}
  $i$ & $0$ & $1$ & $2$ & $3$ & $4$ \\ \hline
  $g_ih_i$ & & $11$ & $00$ & $00$ & $10$ \\ \hline
  $s^{(i)}=s^{(i-1)}\diamond g_ih_i$ & $00$ & $11$ & $11$ & $11$ & $01$ \\ \hline
  $g'_i=\outt(s^{(i-1)},g_ih_i)_1$ & & $1$ & $0$ & $0$ & $0$\\\hline
  $h'_i=\outt(s^{(i-1)},g_ih_i)_2$ & & $1$ & $0$ & $0$ & $1$\\
\end{tabular}
\end{table}

Because the parity keeps track of whether the remaining bits are to be compared
w.r.t.\ the standard or ``reflected'' order, the state machine performs the
comparison correctly w.r.t.\ the meaning of the states indicated in
Figure~\ref{fig:fsm}.
\begin{lemma}\label{lem:comparison}
Let $g,h\in \BB{B}$ and $i\in [B+1]$. Then
\begin{itemize}
  \item $s^{(i)}=00$ is equivalent to $g_{1,i}=h_{1,i}$ and $g\prec h$ if and
  only if $g_{i+1,B}\prec h_{i+1,B}$,
  \item $s^{(i)}=11$ is equivalent to $g_{1,i}=h_{1,i}$ and $g\prec h$ if and
  only if $g_{i+1,B}\succ h_{i+1,B}$,
  \item $s^{(i)}=01$ is equivalent to $g\prec h$, and
  \item $s^{(i)}=10$ is equivalent to $g\succ h$.
\end{itemize}
\end{lemma}

\begin{proof}
We show the claim by induction on $i$. It holds for $i=0$, as $s^{(i)}=00$,
$g_{1,0}=h_{1,0}$ is the empty string, and $g\prec h$ if and only if
$g_{1,B}=g\prec h=h_{1,B}$. For the step from $i-1\in [B]$ to $i$, we make a
case distinction based on $s^{(i-1)}$.
\begin{itemize}
  \item [$s^{(i-1)}=00$:] By the induction hypothesis, $g_{1,i-1}=h_{1,i-1}$ and
  $g\prec h$ if and only if $g_{i,B}\prec h_{i,B}$. Thus, if $g_ih_i=00$,
  $s^{(i)}=00$, $g_{1,i}=h_{1,i}$, and by the recursive definition of the code,
  $g_{i,B}\prec h_{i,B}\Leftrightarrow g_{i+1,B}\prec h_{i+1,B}$. Similarly, if
  $g_ih_i=11$, also $g_{1,i}=h_{1,i}$, but the code for the remaining bits is
  ``reflected,'' i.e., $g\prec h\Leftrightarrow g_{i+1,B}\succ h_{i+1,B}$. If
  $g_ih_i=01$, the definition implies that $g\prec h$ regardless of further
  bits, and if $g_ih_i=10$, $g\succ h$ regardless of further bits.
  \item [$s^{(i-1)}=11$:] Analogously to the previous case, noting that
  reflecting a second time results in the original order.
  \item [$s^{(i-1)}=01$:] By the induction hypothesis, $g\prec h$. As $01$ is an
  absorbing state, also $s^{(i)}=01$.
  \item [$s^{(i-1)}=10$:] By the induction hypothesis, $g\succ h$. As $10$ is an
  absorbing state, also $s^{(i)}=10$.\qedhere
\end{itemize}
\end{proof}

\begin{table}
\centering
\caption{Computing $\rgmax\{g,h\}_i$ and $\rgmin\{g,h\}_i$ from the current
state $s^{(i-1)}$ and inputs $g_i$ and $h_i$.}
\label{tab:output}
\begin{tabular}{|c||c |c |c |c |}
\hline \rule{0pt}{9pt}
$s^{(i-1)}$ & $\rgmax\{g,h\}_i$ & $\rgmin\{g,h\}_i$\\ \hline \hline
00 & $\max\{g_i,h_i\}$ & $\min\{g_i,h_i\}$\\
10 & $g_i$ & $h_i$\\
11 & $\min\{g_i,h_i\}$ & $\max\{g_i,h_i\}$\\
01 & $h_i$ & $g_i$\\
\hline
\end{tabular}
\end{table}
This lemma gives rise to a sequential implementation of $\twosort(B)$ based on
the given state machine, for input strings in $\BB{B}$.
Table~\ref{tab:output} lists the $i^{th}$ output bit as function of $s^{(i-1)}$
and the pair $g_ih_i$. Correctness of this computation follows immediately from
Lemma~\ref{lem:comparison}.
\begin{table}[t]
\centering
\caption{Operators for next state and output. The first operand is the current state,
the second is the next input.}
\label{tab:operators}
\begin{subtable}{.45\columnwidth}
\centering
\caption{The $\diamond$ operator.}\label{tab:diamond}
\begin{tabular}{|c||c |c |c |c |}
\hline
  $\diamond$ & 00 & 01 & 11 & 10\\ \hline \hline
    	  00 & 00 & 01 & 11 & 10 \\
	      01 & 01 & 01 & 01 & 01 \\
  		  11 & 11 & 10 & 00 & 01 \\
		  10 & 10 & 10 & 10 & 10 \\
  \hline
\end{tabular}
\end{subtable}
\begin{subtable}{.5\columnwidth}
\caption{The $\outt$ operator.}
\centering
\label{tab:outt}
\begin{tabular}{|c||c |c |c |c |}
\hline
  $\outt$ & 00 & 01 & 11 & 10\\ \hline \hline
       00 & 00 & 10 & 11 & 10 \\
	   01 & 00 & 10 & 11 & 01 \\
  	   11 & 00 & 01 & 11 & 01 \\
	   10 & 00 & 01 & 11 & 10 \\
  \hline
\end{tabular}
\end{subtable}
\end{table}

We can express the transition function of the state machine as an (as easily
verified) associative operator $\diamond$ taking the current state and input
$g_ih_i$ as argument and returning the new state. Then
$s^{(i)}=s^{(i-1)}\diamond g_ih_i$, where $\diamond$ is given in
Table~\ref{tab:diamond} and $s^{(0)}=00$. The $\outt$ operator is derived from
Table~\ref{tab:output} by evaluating $\rgmax\{g,h\}_i$ and $\rgmin\{g,h\}_i$ for
all possible values of $g_ih_i\in \BB{2}$. Noting that $s^{(0)}\diamond x =
00\diamond x=x$ for all $x\in \BB{2}$, we arrive at the following corollary.
\begin{corollary}\label{cor:comparison}
For all $i\in [1,B]$, we have that
\begin{equation*}
\rgmax\{g,h\}_i\rgmin\{g,h\}_i =
\outt\left(\bigdiamond_{j=1}^{i-1}g_jh_j,g_ih_i\right).
\end{equation*}
\end{corollary}

Our goal in this section is to extend this approach to potentially metastable
inputs.

\subsection{Dealing with Metastable Inputs}\label{sec:optimal}

Our strategy is to replace all involved operators by their metastable closure:
for $i\in [1,B]$ (i) compute $s^{(i)}_{\metas}$, (ii) determine
$\rgmaxM\{g,h\}_i$ and $\rgminM\{g,h\}_i$ according to
Table~\ref{tab:output}, and finally (iii) exploit associativity of the operator
computing the state $s_\metas^{(i)}$ for usage in the PPC framework (\cite{ladner1980parallel},
see Section~\ref{sec:ppc}). Thus, we only need to implement $\diamond_{\metas}$
and the $\outt_\metas$ (both of constant size), plug them into the framework,
and immediately obtain an efficient circuit.

The reader may ask why we compute $s^{(i)}_{\metas}$ for all
$i\in [0,B-1]$ instead of computing only $s^{(B)}_{\metas}$ with a simple tree
of $\diamond_\metas$ elements, which would yield a smaller circuit. Since
$s^{(B)}_{\metas}$ is the result of the comparison of the entire strings, it
could be used to compute all outputs, i.e., we could compute the output by
$\outt_\metas(s^{(B)}_{\metas}, g_ih_i)$ instead of $\outt_\metas(s^{(i-1)}_{\metas},
g_ih_i)$. However, in case of metastability, this may lead to incorrect results.
This can be seen in the example run of the FSM given in
Table~\ref{tab:exampletwo}. We thus compute every intermediate state
$s^{(i)}_{\metas}$.

\begin{table}[t]
  \centering
  \caption{Run of the FSM on inputs $g=0\metas10$ and $h=0010$, showing
  that computing only the last state is insufficient. This yields
  $\outt_\metas(1\metas,\metas0)=\bigstarr\{00,01,10\}=\metas\metas$ as second
  output, but $\outt_\metas(00,\metas0)=\bigstarr\{00,10\}=\metas0$ is correct.}
  \label{tab:exampletwo}
\begin{tabular}{c|c|c|c|c|c}
  $i$ & $0$ & $1$ & $2$ & $3$ & $4$ \\ \hline
  $g_ih_i$ & & $00$ & $\metas0$ & $11$ & $00$ \\ \hline
  $s_\metas^{(i)}=s_\metas^{(i-1)}\diamond_\metas g_ih_i$ & $00$ & $00$ & $\metas0$ & $1\metas$ & $1\metas$ \\ \hline
  $\outt_\metas(s_\metas^{(4)},g_ih_i)$ & & $00$ & $\metas\metas$ & $11$ & $00$\\ \hline
  $\outt_\metas(s_\metas^{(i-1)},g_ih_i)$ & & $00$ & $\metas0$ & $11$ & $00$\\
\end{tabular}
\end{table}

Unfortunately, even with this modification it is not obvious that our approach
yields correct outputs. There are three hurdles to overcome:
\begin{enumerate}
\item [(P1)] Show that $\diamond_{\metas}$ is associative.
\item [(P2)] Show that repeated application of $\diamond_{\metas}$ computes
$s^{(i)}_{\metas}$.
\item [(P3)] Show that applying $\outt_\metas$ to $s^{(i-1)}_\metas$ and
$g_ih_i$ results for all valid strings in $\rgmaxM\{g,h\}_i\rgminM\{g,h\}_i$.
\end{enumerate}

Regarding the first point, we note the statement that $\diamond_{\metas}$ is
associative does not depend on $B$. In other words, it can be verified by
checking for all possible $x,y,z\in \BBM{2}$ whether
$(x\diamond_{\metas}y)\diamond_{\metas}z=x\diamond_{\metas}(y\diamond_{\metas}z)$.
While it is tractable to manually verify all $3^6=729$ cases (exploiting various
symmetries and other properties of the operator), it is tedious and prone to
errors. Instead, we verified that both evaluation orders result in the same
outcome by a short computer program.
\begin{theorem}\label{thm:assoc}
(P1) holds, i.e., $\diamond_{\metas}$ is associative.
\end{theorem}
Apart from being essential for our construction, this theorem simplifies
notation; in the following, we may write
\begin{equation*}
\left(\bigdM\right)_{i=1}^j g_ih_i \coloneqq
g_1h_1\diamond_{\metas}g_2h_2\diamond_{\metas}\ldots \diamond_{\metas}g_jh_j\,,
\end{equation*}
where the order of evaluation does not affect the result.

We stress that in general the closure of an associative operator needs not be
associative. A counter-example is given by binary addition modulo $4$:
\begin{align*}
(0\metas+_{\metas}01)+_{\metas}01 = \metas\metas \neq 1\metas = 0\metas
+_{\metas} (01+_{\metas}01).
\end{align*}

\begin{table}
\caption{The $\diamond_{\metas}$ operator. The first operand is the current
state, the second are the next input bits.}\label{tab:diamondM}
\centering
\begin{tabular}{|c||c |c |c |c |c |c |c |c |c |}
\hline
  $\diamond_{\metas}$ & 00 & 0\metas & 01 & \metas 1 & 11 & 1\metas & 10 & \metas0 & \metas \metas\\
  \hline \hline
  00 & 00 & 0\metas & 01 & \metas 1 & 11 & 1\metas & 10 & \metas0 & \metas \metas\\
  0\metas & 0\metas & 0\metas & 01 & \metas 1 & \metas1 & \metas\metas & \metas\metas & \metas\metas & \metas \metas\\
  01 & 01 & 01 & 01 & 01 & 01 & 01 & 01 & 01 & 01 \\
  \metas 1 & \metas 1 & \metas\metas & \metas\metas & \metas \metas & 0\metas & 0\metas & 01 & \metas1 & \metas \metas\\
  11 & 11 & 1\metas & 10 & \metas 0 & 00 & 0\metas & 01 & \metas1 & \metas \metas\\
  1\metas & 1\metas & 1\metas & 10 & \metas 0 & \metas0 & \metas\metas & \metas\metas & \metas\metas & \metas \metas\\
  10 & 10 & 10 & 10 & 10 & 10 & 10 & 10 & 10 & 10 \\
  \metas 0 & \metas0 & \metas\metas & \metas\metas & \metas \metas & 1\metas & 1\metas & 10 & \metas0 & \metas \metas\\
  \metas\metas & \metas\metas & \metas\metas & \metas\metas & \metas\metas & \metas\metas & \metas\metas & \metas\metas & \metas\metas & \metas\metas\\
  \hline
\end{tabular}
\end{table}

\subsection{Determining $s^{(i)}_{\metas}$}
For convenience of the reader, Table~\ref{tab:diamondM} gives the truth table of
$\diamond_{\metas}\colon\BBM{2}\times\BBM{2}\rightarrow\BBM{2}$. We need to show
that repeated application of this operator to the input pairs $g_jh_j$, $j\in
[1,i]$, actually results in $s^{(i)}_{\metas}$. This is closely related to the
key observation that if in a valid string there is a metastable bit at position
$m$, then the remaining $B-m$ following bits are the maximum codeword of a
$(B-m)$-bit code.

\begin{observation}\label{obs:maxencoding}
For $g\in\validrg{B}$, if there is an index $1\leq m<B$ such that $g_m=\metas$
then $g_{m+1,B}=10^{B-m-1}$.
\end{observation}
\begin{proof}
  List the codewords in order. By the recursive definition of the code, removing
  the first $m-1$ bits of the code leaves us with $2^{m-1}$ repetitions of
  $(B-m+1)$-bit code alternating between listing it in order and in reverse
  (``reflected'') order. Also by the recursive definition, the $m^{th}$ bit
  toggles only when the $(B-m)$-bit code resulting from removing it is at its last
  codeword, $10^{B-m-1}$.
\end{proof}

Our reasoning will be based on distinguishing two main cases: one is that
$s^{(i)}_{\metas}$ contains at most one metastable bit, the other that
$s^{(i)}_{\metas}=\metas \metas$. For each we need a technical statement.
\begin{observation}\label{obs:starkiller}
If $\left|\res\left(s^{(i)}_{\metas}\right)\right|\leq 2$ for any $i\in [B+1]$,
then
$\res(s^{(i)}_{\metas})= \bigdiamond_{j=1}^i \res(g_jh_j)$.
\end{observation}
\begin{proof}
With $S\coloneqq \bigdiamond_{j=1}^i \res(g_jh_j)$, we have that
$\res\left(s^{(i)}_{\metas}\right)=\res(\bigstarr S)$. The claim thus follows
from Observation~\ref{obs:reverseresolution}.
\end{proof}
\begin{lemma}\label{lem:mm}
Suppose that for valid strings $g,h\in \validrg{B}$, it holds that
$s^{(i)}_{\metas}=\metas\metas$ for some $i\in [1,B]$. Then $g=h$ and
$s^{(j)}_{\metas}=\metas\metas$ for all $j\in [i,B]$.
\end{lemma}
\begin{proof}
  Because $R\coloneqq \bigdiamond_{k=1}^i \res(g_kh_k)\subseteq \BB{2}$ and
  $s^{(i)}=\bigstarr R$, it must hold that (i)~$\{00,11\}\subseteq R$, or
  (ii)~$\{01,10\}\subseteq R$. By Lemma~\ref{lem:comparison}, (i) implies that there
  are stabilizations $g',g''\in \res(g_{1,i})$ and $h',h''\in \res(h_{1,i})$ such
  that $g'=h'$, $\parity(g')=0$, $g''=h''$, and $\parity(g'')=1$, while (ii)
  implies such $g',g'',h',h''$ with $g'\prec h'$ and $g''\succ h''$. Checking
  Definition~\ref{def:order} (cf.~Table~\ref{table:validinputs}), we see that both
  options necessitate that $g_{1,i}=h_{1,i}$ with some metastable bit. Denoting by
  $m\in [1,i-1]$ the index such that $g_m=h_m=\metas$.
  Observation~\ref{obs:maxencoding} shows that $g_{m+1,B}=h_{m+1,B}=10^{B-m-1}$.
  In particular, $g=h$, showing (again by Lemma~\ref{lem:comparison} that (i) or
  (ii) (in fact both) also apply to $\bigdiamond_{k=1}^j \res(g_kh_k)$ for any
  $j\in [i,B]$ (cf.~the $00$ and $11$ columns in Table~\ref{tab:diamond}). We
  conclude that $s^{(j)}_{\metas}=\metas \metas$ for any such $j$.
\end{proof}

Equipped with these tools, we are ready to prove the second statement.
\begin{theorem}\label{thm:statemetasi}
(P2) holds, i.e., for all $g,h\in \validrg{B}$ and $i\in[1,B]$,
$s^{(i)}_{\metas} = \left(\bigdM\right)_{j=1}^i g_jh_j$.
\end{theorem}

\begin{proof}
We show the claim by induction on $i$. Trivially, we have that
$s^{(0)}_{\metas}=s^{(0)}=00$ and thus for $i=1$ that
\begin{equation*}
s^{(1)}_{\metas}=s^{(0)}_{\metas}\diamond_{\metas}g_1h_1 =
00\diamond_{\metas}g_1h_1 = g_1h_1 = \left(\bigdM\right)_{j=1}^1 g_1h_1\,.
\end{equation*}

Hence, suppose that the claim has been established for $i-1\in [1,B-1]$ and
consider index $i$. If $\left|\res\left(s^{(i-1)}_{\metas}\right)\right|\leq 2$,
Observation~\ref{obs:starkiller} and the induction hypothesis yield that
\begin{align*}
\left(\bigdM\right)_{j=1}^i g_jh_j
&= \left(\left(\bigdM\right)_{j=1}^{i-1}
g_jh_j\right)\diamond_{\metas}g_ih_i\\
&= s_{\metas}^{(i-1)}\diamond_{\metas}g_ih_i
= \bigstarr\! \left(\!\res\!\left(s_{\metas}^{(i-1)}\right)\diamond
\res(g_ih_i)\!\right)\\
&= \bigstarr \bigdiamond_{j=1}^i\res(g_ih_i)= s^{(i)}_{\metas}\,.
\end{align*}

It remains to consider the case that $s^{(i-1)}_{\metas}=\metas \metas$.
By Lemma~\ref{lem:mm}, $s^{(i)}_{\metas}=\metas\metas$, too.
Thus,
\begin{equation*}
\left(\bigdM\right)_{j=1}^i g_jh_j
= s_{\metas}^{(i-1)}\diamond_{\metas}g_ih_i
= \metas\metas \diamond_{\metas}g_ih_i
=\metas \metas
=s^{(i)}_{\metas}.\qedhere
\end{equation*}
\end{proof}

\begin{table}[t]
\centering
\caption{The $\outt_{\metas}$ operator. The first operand is the current state, the second is the next input
bits.\label{tab:outtM}}
\begin{tabular}{|c||c |c |c |c |c |c |c |c |c |}
\hline
  $\outt_{\metas}$ & 00 & 0\metas & 01 & \metas1 & 11 & 1\metas & 10 & \metas0 &
  \metas\metas\\ \hline \hline 00 & 00 & \metas0 & 10 & 1\metas & 11 & 1\metas & 10 & \metas0 & \metas\metas \\
  0\metas & 00 & \metas0 & 10 & 1\metas & 11 & \metas\metas & \metas\metas & \metas\metas & \metas\metas \\
  01 & 00 & \metas0 & 10 & 1\metas & 11 & \metas 1 & 01 & 0\metas & \metas\metas \\
  \metas 1 & 00 & \metas\metas & \metas\metas & \metas\metas & 11 & \metas 1 & 01 & 0\metas & \metas\metas \\
  11 & 00 & 0\metas & 01 & \metas1 & 11 & \metas1 & 01 & 0\metas & \metas\metas \\
  1\metas & 00 & 0\metas & 01 & \metas1 & 11 & \metas\metas & \metas \metas & \metas\metas & \metas\metas \\
  10 & 00 & 0\metas & 01 & \metas1 & 11 & 1\metas & 10 & \metas0 & \metas\metas \\
  \metas 0 & 00 & \metas\metas & \metas\metas & \metas\metas & 11 & 1\metas & 10 & 0\metas & \metas\metas \\
  \metas\metas & 00 & \metas\metas & \metas\metas & \metas\metas & 11 & \metas\metas & \metas\metas & \metas\metas & \metas\metas \\
  \hline
\end{tabular}
\end{table}

\subsection{Obtaining the Outputs from $s^{(i)}_{\metas}$}\label{sec:outputdet}
Recall that $\outt\colon \BB{2}\times \BB{2}\to \BB{2}$ is the
operator given in Table~\ref{tab:output} computing $\rgmax\{g,h\}_i
\rgmin\{g,h\}_i$ out of $s^{(i-1)}$ and $g_ih_i$. For convenience of the reader,
we provide the truth table of $\outt_{\metas}\colon\BBM{2}\times\BBM{2}
\rightarrow\BBM{2}$ in Table~\ref{tab:outtM}. We derive the third property.
\begin{theorem}\label{thm:out}
(P3) holds, i.e., given valid inputs $g,h\in \validrg{B}$ and $i\in [1,B]$,
$\outt_{\metas}(s^{(i-1)}_{\metas},g_ih_i) = \rgmaxM\{g,h\}_i\rgminM\{g,h\}_i$.
\end{theorem}

\begin{proof}
Assume first that $\left|\res\left(s^{(i-1)}_{\metas}\right)\right|\leq 2$. Then
\begin{align*}
&\,\outt_{\metas}(s^{(i-1)}_{\metas}(g,h),g_ih_i)\\
=\,& \bigstarr
\outt\left(\res\left(s^{(i-1)}_{\metas}(g,h)\right),\res(g_ih_i)\right)\\
\stackrel{\mathrm{Obs.\,}\ref{obs:starkiller}}{=}\,& \bigstarr
\outt\left(\bigdiamond_{j=1}^{i-1}\res(g_jh_j),\res(g_ih_i)\right)\\
\stackrel{\mathrm{Cor.\,}\ref{cor:comparison}}{=}\,& \bigstarr\left(
\rgmax\{\res(g),\res(h)\}_i \rgmin\{\res(g),\res(h)\}_i\right)\\
=\,& \rgmaxM\{g,h\}_i\rgminM\{g,h\}_i\,.
\end{align*}
Otherwise, $s^{(i-1)}_{\metas}=\metas\metas$. Then, by Lemma~\ref{lem:mm},
$g=h$. In particular, $g_i=h_i$. Checking Table~\ref{tab:outtM}, we see that for
all $b\in \BBM{}$, it holds that $\outt_{\metas}(\metas\metas,bb)=bb$. Therefore,
\begin{equation*}
\outt_{\metas}(s^{(i-1)}_{\metas}(g,h),g_ih_i)=g_ih_i=\rgmaxM\{g,h\}_i\rgmin\{g,h\}_i
\end{equation*}
in this case as well.
\end{proof}

\section{The PPC Framework}\label{sec:ppc}

In order to derive a small circuit from the results of
Section~\ref{sec:operators}, a straightforward approach would be to unroll the
FSM. We could design a circuit implementing the transition function
$\diamond_\metas$ and apply it $B$ times to the starting state $s^{(0)}$ and
each input $g_ih_i$. However, computing the sequence of states step by step
yields a (non-optimal) linear depth of at least $B$.

Hence, we make use of the PPC framework by Ladner and
Fischer~\cite{ladner1980parallel}. They describe a generic method that is
applicable to \emph{any} finite state machine translating a sequence of $B$
input symbols to $B$ output symbols, to obtain circuits of size $\BO(B)$ and
depth $\BO(\log B)$. They observe that each input symbol defines a restricted
transition function. Compositions of these functions evaluated on the starting
state yield the state of the machine after receiving corresponding inputs. The
major advantage of the technique is that compositions of restricted transition
functions can be computed in parallel due to associativity, yielding a depth of
$\BO(\log B)$. This matches our needs, as we need to determine
$s^{(i)}_{\metas}$ for each $i\in [B]$. However, their generic construction
involves large constants. Fortunately, we have established that
$\diamond_{\metas}\colon \BBM{2}\times \BBM{2}\rightarrow\BBM{2}$ is an
associative operator,
permitting us to directly apply the circuit templates for associative operators
they provide for computing $s^{(i)}_{\metas}=\left(\bigdM\right)_{j=1}^i g_jh_j$
for all $i\in [B]$. Accordingly, we discuss these templates only. During
discussion of the basic construction we show a minor improvement on their
results.

Before proceeding, the reader may want to take a look at the example given in
Figure~\ref{fig:completeex}, which shows how a $\twosort(9)$ derived from our
construction processes an input pair.

\begin{figure}
\begin{center}
\input{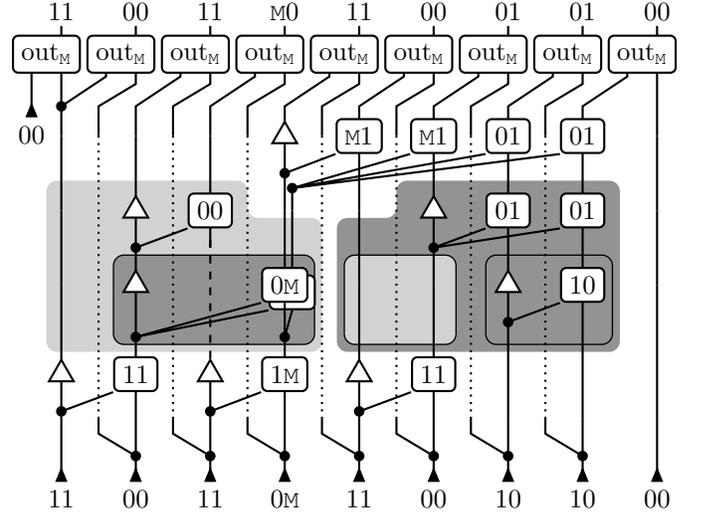}
\end{center}
\caption{An example for a computation of the $\twosort(9)$ circuit arising from
our construction for fan-out $f=3$. The inputs are $g=101010110$ and
$h=101\metas10000$; see Table~\ref{tab:exrun} for $s^{(i)}_{\metas}(g,h)$ and
the output. We labeled each $\diamond_\metas$ by its output. Buffers and
duplicated gates (here the one computing $0\metas$) reduce fan-out, but do not
affect the computation. Grey boxes indicate recursive steps of the PPC
construction; see also Figure~\ref{fig:treeex} for a larger PPC circuit using
the one here in its ``right'' top-level recursion. For better readability, wires
not taking part in a recursive step are dashed or dotted.}
\label{fig:completeex}
\end{figure}

\begin{table}
  \centering
  \caption{Example run of the FSM in Figure~\ref{fig:fsm} on inputs
  $g=101010110$ and $h=101\metas10000$. We drop $s_\metas^{(9)}$, as it is not
  needed to compute $g'_9h'_9$.}
  \label{tab:exrun}
  \begin{tabular}{c|c|c|c|c|c|c|c|c|c|c}
    $i$ & $0$ & $1$ & $2$ & $3$ & $4$ & $5$ & $6$ & $7$ & $8$ & $9$ \\ \hline
    $g_ih_i$ & & $11$ & $00$ & $11$ & $0\metas$ & $11$ & $00$ & $10$ & $10$ & $00$ \\ \hline
    $s_\metas^{(i)}$
    & $00$ & $11$ & $11$ & $00$ & $0\metas$ & $\metas1$ & $\metas1$ & $01$ & $01$ & \\ \hline
    $g'_ih'_i$ & & $11$ & $00$ & $11$ & $\metas0$ & $11$ & $00$ & $01$ & $01$ & $00$\\
  \end{tabular}
\end{table}


\subsection{The Basic Construction}
We revisit the templates for parallel computation of all prefixes, i.e., the
part of the framework relevant to our construction. To this end, recall
Definition~\ref{def:ppc}. In our case, $\OP=\diamond_{\metas}$ and $D=\BBM{2}$.
\cite{ladner1980parallel} provides a family of recursive constructions of
$\ppc_{\OP}$ circuits. They are obtained by combining two different recursive
patterns. The first pattern, which optimizes for size of the resulting circuits,
is depicted in Figure~\ref{subfig:pattern1}. We distinguish between even and odd
number of inputs. If $B$ is even, we discard the rightmost gray wire and set
$\bar{B} \coloneqq B$; if $B$ is odd, we set $\bar{B}\coloneqq B-1$ and include
the rightmost wire. In the following, denote by $|C|$ the size of a circuit $C$
and by $d(C)$ its depth.
\begin{lemma}\label{lem:Rsize}
Suppose that $C$ and $P$ are circuits implementing $\OP$ and
$\ppc_{\OP}(\lceil B/2\rceil)$ for some $B\in \NN$, respectively.
Then applying the recursive pattern given at the left of Figure~\ref{fig:tree}
yields a $\ppc_{\OP}(B)$ circuit. It has depth $2d(C)+d(P)$ and size at most
$(B-1)|C|+|P|$. Moreover, the last output is at depth at most $d(C)+d(P)$ of the
circuit.
\end{lemma}

\begin{proof}
Observe that $P$ receives as inputs $d_{2i-1}\OP d_{2i}$ for $i\in [1,\lfloor
B/2\rfloor]$, and in addition $d_B$ in case $B$ is odd. Thus, it outputs
$\pi_i'=\bigoplus_{j=1}^{2i} d_j$ for $i\in [1,\lfloor B/2\rfloor]$, and also
$\pi_{\lceil B/2\rceil}'=\bigoplus_{j=1}^B d_j$ if $B$ is odd. Hence, the
circuit outputs $\pi_i=\bigoplus_{j=1}^i d_j$ if $i\in [1,B]$ is even and
\begin{equation*}
\pi_i=\pi_{i-1}\oplus d_i=\left(\bigoplus_{j=1}^{2\lfloor i/2\rfloor}
d_j\right)\OP d_i = \bigoplus_{j=1}^i d_j
\end{equation*}
if $i\in [1,B]$ is odd, showing correctness. The depth of the circuit is
immediate from the construction, and the size follows from the fact that there
is exactly one instance of $C$ for each even $i\in [1,B]$ before $P$ and one for
each odd $i\in [1,B]\setminus \{1,B\}$ after $P$. Output $\pi_B$ has a depth
that is smaller by $d(C)$, as it is an output of $P$.
\end{proof}

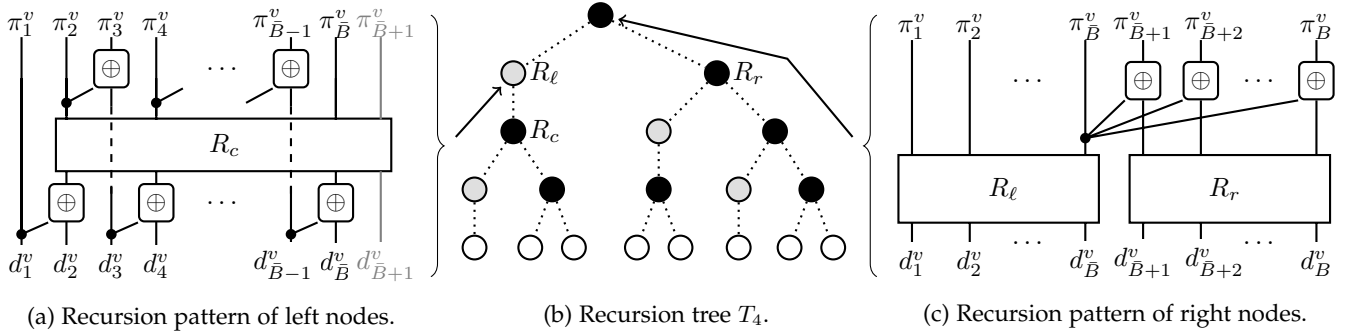
\begin{figure*}
\begin{center}
\begin{subfigure}{.31\textwidth}
\begin{tikzpicture}[scale=0.995]
\draw[thick]
node[] (buff1) {}
++(right:.6) node[] (buff2) {}
++(right:.6) node[operator] (op1) {$\OP$}
++(right:.6) node[] (buff3) {}
++(right:.6) node[operator, draw=none, text=white] (op2) {$\OP$}
++(right:.6) node (buff4) {}
++(right:.6) node[operator] (op3) {$\OP$}
++(right:.6) node[] (buff5) {}
++(right:.6) node[] (buff6) {}

(buff1) ++(up:12pt) -- ++(down:75pt) node[circ] (in1) {} -- ++(down:12pt) node (l1) {}
(buff2) ++(up:12pt) -- ++(down:30pt) ++(down:5pt) node (in2) {}
(op1) -- ++(up:12pt) ++(down:30pt) node (dash1s) {} ++(down:27pt) node(dash1e) {} -- ++(down:18pt) node[circ] (in3) {} -- ++(down:12pt) node (l3) {}
(buff3) ++(up:12pt) -- ++(down:30pt) ++(down:5pt) node (in4) {}
(buff4) ++(up:12pt) ++(down:35pt) node (in6) {}
(op3) -- ++(up:12pt) ++(down:30pt) node (dash2s) {} ++(down:27pt) node(dash2e) {} -- ++(down:18pt) node[circ] (in7) {} -- ++(down:12pt) node (l7) {}
(buff5) ++(up:12pt) -- ++(down:30pt) ++(down:5pt) node (in8) {}
(buff6) ++(down:23pt) node (in9) {}

(buff1) -- (in1)
(buff2) -- (in2) ++(down:16pt) -- ++(down:12pt) node[operator, fill=white] (op4) {$\OP$} -- ++(down:24pt) node (l2) {}
(op1) -- (dash1s)
(in2) ++(up:10pt) node[circ] {} -- (op1.north west)
(buff3) -- (in4) ++(down:16pt) -- ++(down:12pt) node[operator, fill=white] (op5) {$\OP$} -- ++(down:24pt) node (l4) {}
(in4) ++(up:10pt) node[circ] {} -- (op2.north west)
(op3) -- (dash2s)
(in6) ++(up:10pt) -- (op3.north west)
(buff5) -- (in8) ++(down:16pt) -- ++(down:12pt) node[operator, fill=white] (op6) {$\OP$} -- ++(down:24pt) node (l8) {}

(in1) -- (op4.north west)
(in3) -- (op5.north west)
(in7) -- (op6.north west)

(in2) ++(up:4pt)++(left:4pt) node (Rcleft) {}
(in9) ++(down:16pt)++(right:4pt) node (Rcright) {}
(in4) ++(down:6pt)++(right:26pt) node (label) {$R_c$}
(Rcleft) rectangle (Rcright)
;

\draw[thick]
(buff3) ++(right:25pt) node {$\hdots$}
(op5) ++(right:25pt) node {$\hdots$}
;

\draw[dashed, thick]
(dash1s.north)--(dash1e.south)
(dash2s.north)--(dash2e.south)
;

\draw[thick]
(l1) node[fill=white] {$\inp^v_1$}
(l2) node[fill=white] {$\inp^v_2$}
(l3) node[fill=white] {$\inp^v_3$}
(l4) node[fill=white] {$\inp^v_4$}
(l8) node[fill=white] {$\inp^v_{\bar{B}}$}
(l7)++(left:3pt) node[fill=white] {$\inp^v_{\bar{B}-1}$}
(buff1) ++(up:17pt) node[] {$\outp^v_1$}
(buff2) ++(up:17pt) node[] {$\outp^v_2$}
(op1) ++(up:17pt) node[] {$\outp^v_3$}
(buff3) ++(up:17pt) node[] {$\outp^v_4$}
(buff5) ++(up:17pt) node[] {$\outp^v_{\bar{B}}$}
(op3) ++(up:17pt)++(left:3pt) node[] {$\outp^v_{\bar{B}-1}$}
;

\draw[thick, gray]
(buff6) ++(up:12pt) -- (in9) ++(down:16pt) -- ++(down:36pt) node (l9) {}
(buff6) ++(up:17pt)++(right:2pt) node[] {$\outp^v_{\bar{B}+1}$}
(l9)++(right:2pt) node[fill=white] {$\inp^v_{\bar{B}+1}$}
;

\end{tikzpicture}
\caption{Recursion pattern of left nodes.}
\label{subfig:pattern1}
\end{subfigure}
\begin{subfigure}{.33\textwidth}
\begin{tikzpicture}[scale=1.1]
\draw[thick, level/.style={sibling distance=80/#1, level distance=20pt}, every node/.style={circle, solid, draw}, dotted, align=center]
node[fill=black] (root) {}
  child[sibling distance=60pt]{ node[fill=gray!25, label={[label distance=-5pt]1:$R_{\ell}$}] (child1) {}
    child{ node[fill=black, label={[label distance=-5pt]1:$R_c$}] (child11) {}
      child{ node[fill=gray!25] (child111) {}
        child{ node (child1111) {}}
      }
      child{ node[fill=black] (child112) {}
        child[sibling distance=15pt]{ node (child1121) {}}
        child[sibling distance=15pt]{ node (child1122) {}}
      }
    }
  }
  child{ node[fill=black, label={[label distance=-5pt]1:$R_r$}] (child2) {}
    child[sibling distance=40pt]{ node[fill=gray!25] (child21) {}
      child{ node[fill=black] (child211) {}
        child[sibling distance=15pt]{ node (child2111) {}}
        child[sibling distance=15pt]{ node (child2112) {}}
      }
    }
    child[sibling distance=40pt]{ node[fill=black] (child22) {}
      child[sibling distance=25pt]{ node[fill=gray!25] (child221) {}
        child{ node (child2211) {}}
      }
      child[sibling distance=25pt]{ node[fill=black] (child222) {}
        child[sibling distance=15pt]{ node (child2221) {}}
        child[sibling distance=15pt]{ node (child2222) {}}
      }
    }
  }
;

\draw decorate[decoration={brace,mirror,raise=0pt, amplitude=5pt}]{
(child1111) ++(left:15pt)++(down:10pt) --node[right=3pt] (a1) {} ++(up:90pt)};
\draw decorate[decoration={brace,raise=0pt, amplitude=5pt}]{
(child2222) ++(right:15pt)++(down:10pt) --node[left=3pt] (a2) {} ++(up:90pt)};
\draw[thick, ->, shorten >=2pt] (a1) -- (child1);
\draw[thick, ->, shorten >=2pt] (a2) -- ++(-25pt,30pt) -- (root);
\node (spacing) at (1,-3.2) {};
\end{tikzpicture}
\caption{Recursion tree $T_4$.}
\label{subfig:rectree}
\end{subfigure}
\begin{subfigure}{.33\textwidth}
\begin{tikzpicture}[scale=1.28]
\draw[thick]
node (out1) {}
++(right:.6) node (out2) {}
++(right:.6) node (out3) {}
++(right:.6) node (out4) {}
++(right:.6) node[operator] (op1) {$\OP$}
++(right:.6) node[operator] (op2) {$\OP$}
++(right:.6) node (op3) {}
++(right:.6) node[operator] (op4) {$\OP$}

(op1) -- ++(up:12pt)
(op2) -- ++(up:12pt)
(op4) -- ++(up:12pt)

(out1) ++(up:12pt) -- ++(down:34pt) ++(down:20pt) -- ++(down:12pt) node (in1) {}
(out2) ++(up:12pt) -- ++(down:34pt) ++(down:20pt) -- ++(down:12pt) node (in2) {}
(out3) ++(down:48pt) node (in3) {}
(out4) ++(up:12pt) -- ++(down:34pt) ++(down:20pt) -- ++(down:12pt) node (in4) {}
(op1) -- ++(down:22pt) ++(down:20pt) -- ++(down:12pt) node (in5) {}
(op2) -- ++(down:22pt) ++(down:20pt) -- ++(down:12pt) node (in6) {}
(op3) ++(down:48pt) node (in7) {}
(op4) -- ++(down:22pt) ++(down:20pt) -- ++(down:12pt) node (in8) {}

(out4) ++(down:17pt) node[circ, fill=black, draw] {} -- (op1.north west)
(out4) ++(down:17pt) -- (op2.north west)
(out4) ++(down:17pt) -- (op4.north west)

(out1) ++(down:22pt)++(left:4pt) node (Rrleft) {}
(out4) ++(down:42pt)++(right:4pt) node (Rrright) {}
(out2) ++(down:32pt)++(right:10pt) node (label1) {$R_{\ell}$}
(Rrleft) rectangle (Rrright)

(op1) ++(down:22pt)++(left:4pt) node (Rlleft) {}
(op4) ++(down:42pt)++(right:4pt) node (Rlright) {}
(op3) ++(down:32pt)++(left:10pt) node (label2) {$R_r$}
(Rlleft) rectangle (Rlright)
;

\draw[thick]
(op3) node {$\hdots$}
(out3) node {$\hdots$}
(in3) node {$\hdots$}
(in7) node {$\hdots$}

(in1) node[fill=white] {$\inp^v_1$}
(in2) node[fill=white] {$\inp^v_2$}
(in4) node[fill=white] {$\inp^v_{\bar{B}}$}
(in6) node[fill=white, text=white] {$\inp^v_{\bar{B}+2}$}
(in5) node[fill=white] {$\inp^v_{\bar{B}+1}$}
(in6)++(right:4pt) node {$\inp^v_{\bar{B}+2}$}
(in8) node[fill=white] {$\inp^v_B$}

(out1) ++(up:16pt) node[] {$\outp^v_1$}
(out2) ++(up:16pt) node[] {$\outp^v_2$}
(out4) ++(up:16pt) node[] {$\outp^v_{\bar{B}}$}
(op1) ++(up:16pt) node[] {$\outp^v_{\bar{B}+1}$}
(op2) ++(up:16pt)++(right:4pt) node[] {$\outp^v_{\bar{B}+2}$}
(op4) ++(up:16pt) node[] {$\outp^v_{B}$}
;

\end{tikzpicture}
\caption{Recursion pattern of right nodes.}
\label{subfig:pattern2}
\end{subfigure}
\end{center}
\caption{The recursion tree $T_4$ (center). Right nodes are depicted black, left
nodes gray, and leaves white. The recursive patterns applied
at left and right nodes are shown on the left and right, respectively. At the
root and its left child, we have that $\bar{B}=B/2$; for other nodes, $\bar{B}$
gets halved for each step further down the tree (where the leaves simply wire
their single input to their single output). The left pattern comes in different
variants. The gray wire with
index $\bar{B}+1$ is present only if $B$ is odd; this never occurs in
$\ppc(C,T_b)$, but becomes relevant when initially applying the left pattern
exclusively for $k\in \NN$ steps (see Theorem~\ref{thm:ppc}), reducing the size
of the resulting circuit at the expense of increasing its depth by $k$.}
\label{fig:tree}
\end{figure*}

The second recursive pattern, shown in Figure~\ref{subfig:pattern2}, avoids to
increase the depth of the circuit beyond the necessary $d(C)$ for each level of
recursion. Assume for now that $B$ is a power of $2$. We represent the recursion
as a tree $T_b$, where $b\coloneqq\log B$, given in the center of
Figure~\ref{fig:tree}. It has depth $b$ with all leaves (filled in white) in this depth, and
there are two types of non-leaf nodes: \emph{right} nodes (filled in black) have
two children, a left and a right node, whereas \emph{left} nodes (filled in
gray) have a single child, which is a right node. $T_b$ is essentially a
Fibonacci tree in disguise.
\begin{definition}\label{def:fibonacci}
$T_0$ is a single leaf. $T_1$ consists of the (right) root and two attached
leaves. For $b\geq 2$, $T_b$ can be constructed from $T_{b-1}$ and $T_{b-2}$ by
taking a (right) root $r$, attaching the root of $T_{b-1}$ as its right child, a
new left node $\ell$ as the left child of $r$, and then attaching the root of
$T_{b-2}$ as (only) child of $\ell$.
\end{definition}

The recursive construction is now defined as follows. A right node applies the
pattern given in Figure~\ref{fig:tree} to the right. $R_{\ell}$ is the
circuit (recursively) defined by the subtree rooted at the left child and $R_r$ is
the circuit (recursively) defined by the subtree rooted at the right child.
Furthermore, $\bar{B}=2^{b-d-1}$, where $d\in [b]$ is the depth of the node.
A left child applies the pattern on the left. $R_c$ is (recursively) defined by
the subtree rooted at its child and $\bar{B}=2^{b-d}$, where
$d\in [b]$ is the depth of the node.

The base case for a single input and output
is simply a wire connecting the input to the output, for both patterns. As
$b=\log B$ and each recursive step cuts the number of inputs and outputs in
half, the base case applies if and only if the node is a leaf. Note that the
figure shows the recursive patterns at the root and its left child, where
$\bar{B}=2^{b-1}$ is always even (i.e., in this recursive pattern, the gray wire
with index $\bar{B}+1$ is never present); when applying the patterns to nodes
further down the tree, $\bar{B}$ and $B$ are scaled down by a factor of $2$ for
every step towards the leaves.

In the following, denote by $\ppc(C,T_b)$ the circuit that results from
applying the recursive construction described above to the base circuit $C$
implementing $\oplus$. Moreover, we refer to the $i^{th}$ input and output of
the subcircuit corresponding to node $v\in T_b$ as $d^v_i$ and $\pi^v_i$,
respectively.

\begin{lemma}\label{lem:rdepth}
If $C$ implements $\oplus$, $\ppc(C,T_b)$ is a $\ppc_{\oplus}(2^b)$
circuit.
\end{lemma}
\begin{proof}
We show the claim by induction on $b$. For $b=0$, the circuit correctly wires
the input to the output, as we have only one leaf. For $b=1$, the first output
equals the first input and the
second output is the result of feeding both inputs into a copy of $C$.

For $b\geq 2$, by the induction hypothesis the circuit $R_c$ used in the
construction at the left child of the root is a $\ppc_{\oplus}(2^{b-2})$
circuit. By Lemma~\ref{lem:Rsize}, the circuit $R_{\ell}$ in the construction at
the root is thus a $\ppc_{\oplus}(2^{b-1})$ circuit, showing that it outputs
$\pi_i^{\ell} = \bigoplus_{j=1}^i d_j = \pi_i$ for all $i\in [1,2^{b-1}]$. From
the induction hypothesis for $b-1$, we get that the circuit $R_r$ used in the
construction at the root is a $\ppc_{\oplus}(2^{b-1})$ circuit, showing that it
outputs $\pi_i^r = \bigoplus_{j=2^{b-1}+1}^{2^{b-1}+i} d_j$ for all $i\in
[1,2^{b-1}]$. By construction of the right recursion pattern we conclude that
for $i\in [2^{b-1}+1,2^b]$, we get the outputs
$\pi_{2^{b-1}} \oplus \pi_{i-2^{b-1}}^r = \bigoplus_{j=1}^i d_i = \pi_i$.
\end{proof}

\begin{lemma}\label{lem:depth}
$\ppc(C,T_b)$ has depth $b\cdot d(C)$.
\end{lemma}
\begin{proof}
We prove the claim by induction on $b$; it trivially holds for $b=0$, as we have
only one leaf. For $b=1$, $T_b$ is a right node with two leaves. The
two leaves have depth $0$; clearly, applying the right
pattern from Figure~\ref{fig:tree} then results in depth $d(C)$. For $b\geq 2$,
the subcircuit $R_r$ at the root has depth $(b-1)\cdot d(C)$ by the induction
hypothesis. For the subcircuit $R_{\ell}$ at the root, consider its subcircuit
$R_c$. By the induction hypothesis it has depth $(b-2)\cdot
d(C)$. Hence, by Lemma~\ref{lem:Rsize}, $R_{\ell}$ has depth $b\cdot d(C)$, but
its rightmost output $\pi_{2^{b-1}}^{\ell}$ has depth only $(b-1)\cdot d(C)$.
Thus, by construction the root's circuit has depth $b\cdot d(C)$.
\end{proof}

It remains to bound the size of the circuit. Denote by $F_i$, $i\in \NN$, the
$i^{th}$ Fibonacci number, i.e., $F_1=F_2=1$ and $F_{i+1}=F_i+F_{i-1}$ for all
$2\leq i\in \NN$.
\begin{lemma}\label{lem:size}
$\ppc(C,T_b)$ has size $(2^{b+2}-F_{b+5}+1)|C|$.
\end{lemma}
\begin{proof}
  Denote by $s_b$ the number of copies of $|C|$ in the circuit $\ppc(C,T_b)$. We
  show by induction that $s_b= 2^{b+2}-F_{b+5}+1$ for all $b\in \NN_0$. We have that
  $s_0=0=2^2-F_5+1$ and that $s_1=1=2^3-F_6+1$. For $b\geq 2$, we have that
  $s_b=s_{b-1}+s_{b-2}+s_r+s_{\ell}$, where $s_r$ and $s_{\ell}$ denote the size
  contribution of the recursive steps at the root and its left child,
  respectively. Checking the recursive patterns in Figure~\ref{fig:tree}, we see
  that $s_r=B-\bar{B}=2^{b-1}$ and $s_{\ell}=\bar{B}-1=2^{b-1}-1$. Thus, $s_b=
  s_{b-1}+s_{b-2}+2^b-1$, which the induction hypothesis yields
  \begin{align*}
  s_b &= 2^{b+1}-F_{b+4}+1+2^b-F_{b+3}+1+2^b-1\\
  &=2^{b+2}-F_{b+5}+1\,.\qedhere
  \end{align*}
\end{proof}

Asymptotically, the subtractive term of $F_{b+5}$ is negligible, as
$F_{b+5}\in (1/\sqrt{5}+o(1))((1+\sqrt{5})/2)^{b+5}\subseteq \BO(1.62^b)$;
however, unless $B$ is large, the difference is substantial. We also get a
simple upper bound for arbitrary values of $B$. To this end, we ``split'' in the
recursion such that the left branch is ``complete'' (i.e.\ the number of inputs
is a power of $2$), while applying the same
splitting strategy on the right. This is where our construction differs from and
improves on~\cite{ladner1980parallel}. They perform a balanced split and obtain
an upper bound of $4B$ on the circuit size.
\begin{corollary}\label{cor:rdepth}
For $B\in \NN$ and circuit $C$ implementing $\oplus$, set $b\coloneqq \lceil\log
B\rceil$. Then a $\ppc_{\oplus}(B)$ of depth $\lceil \log B\rceil d(C)$ and size
smaller than $(5B-2^b-F_{b+3})|C|\leq(4B-F_{b+3})$ exists.
\end{corollary}
\begin{proof}
  If $B$ is a power of $2$, the claim follows from
  Lemmas~\ref{lem:rdepth},~\ref{lem:depth}, and~\ref{lem:size}. In particular, for
  $B=1$ and $B=2$, respectively, $\ppc_{\oplus}(C,T_0)$ and $\ppc_{\oplus}(C,T_1)$
  meet the requirements. For $B>2$ that is not a power of $2$, set $b\coloneqq
  \lceil \log B\rceil$ and perform the same construction as for $\ppc(C,T_b)$, but
  replace $R_r$ at the root by the (recursively given) $\ppc_{\oplus}(B-2^{b-1})$
  circuit.

  Correctness is immediate from the recursive construction and
  Lemma~\ref{lem:rdepth}. Similarly, the depth bound follows from
  Lemma~\ref{lem:depth} and the recursive construction. Regarding size, we show by
  induction that $s_B$, the number of copies of $C$ required for a circuit for
  $B$ inputs, satisfies the claimed bound. This is already established for the
  base cases of $B=1$ and $B=2$. For $B>3$, we apply Lemma~\ref{lem:Rsize} to
  $R_{\ell}$ in the root's circuit and Lemma~\ref{lem:size} to its subcircuit
  $R_c$, while applying the induction hypothesis to the subcircuit $R_r$ in the
  root's circuit. We get that
  \begin{align*}
  s_B&< s_{2^{b-2}} + |\ppc_{\oplus}(C,T_{b-1})|+ (B-1)\\
  & < \left(2^b-F_{b+3}+1+4\left(B-2^{b-1}\right)+B-1\right)\\
  & = \left(5B-2^b-F_{b+3}\right)\,.\qedhere
  \end{align*}
\end{proof}

We remark that one can give more precise bounds by making case distinctions
regarding the right recursion, which for the sake of brevity we omit here.
Instead, we computed the exact numbers for $B\leq 70$, see
Figure~\ref{fig:improv}.

The construction derived from iterative application of Lemma~\ref{lem:Rsize} can
be combined with $\ppc(C,T_b)$, achieving the following trade-off; note that if
$B=2^b$ for $b\in \NN$, then $F_{\lceil \log B\rceil -k+3}$ can be replaced by
$F_{b-k +5}$.
\begin{theorem}[improving on~\cite{ladner1980parallel}]\label{thm:ppc}
Suppose $C$ implements $\OP$. For all $k\in [0, \lceil \log B\rceil]$ and
$B\in \NN$, there is a $\ppc_{\OP}(B)$ circuit of depth $(\lceil \log
B\rceil+k)d(C)$ and size at most
\begin{equation*}
\left(\left(2+\frac{1}{2^{k-1}}\right)B-F_{\lceil \log B\rceil -k
+3}\right)|C|\,.
\end{equation*}
\end{theorem}
\begin{proof}
  For $k$ steps, we apply Lemma~\ref{lem:Rsize}, where in the final recursive step
  we use the circuit from Corollary~\ref{cor:rdepth}. Correctness is immediate
  from Lemma~\ref{lem:Rsize} and~Corollary~\ref{cor:rdepth}.

  Denote by $B_i$, $i\in [k+1]$, the number of in- and outputs of the (sub)circuit
  at depth $i$ of the recursion, i.e., $B_0=B$ and $B_{i+1}=\lceil B_i/2\rceil$
  for all $i\in [k]$. We have that $B_i\leq B/2^i+\sum_{j=1}^i
  2^{-j}<B/2^i+1$, which follows inductively via
  \begin{align*}
  B_{i+1}=\left\lceil \frac{B_i}{2}\right\rceil
  &\leq \left\lceil \frac{B_0/2^i+\sum_{j=1}^i 2^{-j}}{2}\right\rceil\\
  &\leq \frac{B_0}{2^{i+1}}+\left(\sum_{j=1}^i 2^{-j-1}\right)+\frac{1}{2}\\
  &=\frac{B}{2^{i+1}}+\sum_{j=1}^{i+1} 2^{-j}\,.
  \end{align*}
  Observe that $\lceil \log B_{i+1}\rceil=\lceil \log
  B_i\rceil-1$ for all $i\in [k]$. By Lemma~\ref{lem:Rsize}
  and~Corollary~\ref{cor:rdepth}, the size of the resulting circuit is thus
  bounded by
  \begin{align*}
  &\left(\frac{4B}{2^k}-F_{\lceil \log B\rceil -k +3}
  +\sum_{i=0}^{k-1}(B_i-1)\right)|C|\\
  <&\left(\frac{4B}{2^k}-F_{\lceil \log B\rceil -k
  +3}+\sum_{i=0}^{k-1}\frac{B}{2^i}\right)|C|\\
  =&\left(\left(\frac{2}{2^{k-1}}+2-\frac{1}{2^{k-1}}\right)B-F_{\lceil \log
  B\rceil -k +3}\right)|C|\\
  =&\;\left(\left(2+\frac{1}{2^{k-1}}\right)B-F_{\lceil \log B\rceil -k
  +3}\right)|C|\,.
  \end{align*}
  Finally, Lemmas~\ref{lem:Rsize} and~\ref{lem:depth} show that the circuit has
  depth
  \begin{equation*}
  (2k+\lceil \log B_k\rceil)d(C) = (\lceil \log B\rceil + k)d(C)\,.\qedhere
  \end{equation*}
\end{proof}

\begin{figure}
\centering
\includegraphics[width=\linewidth, clip, trim=4 29 9 35]{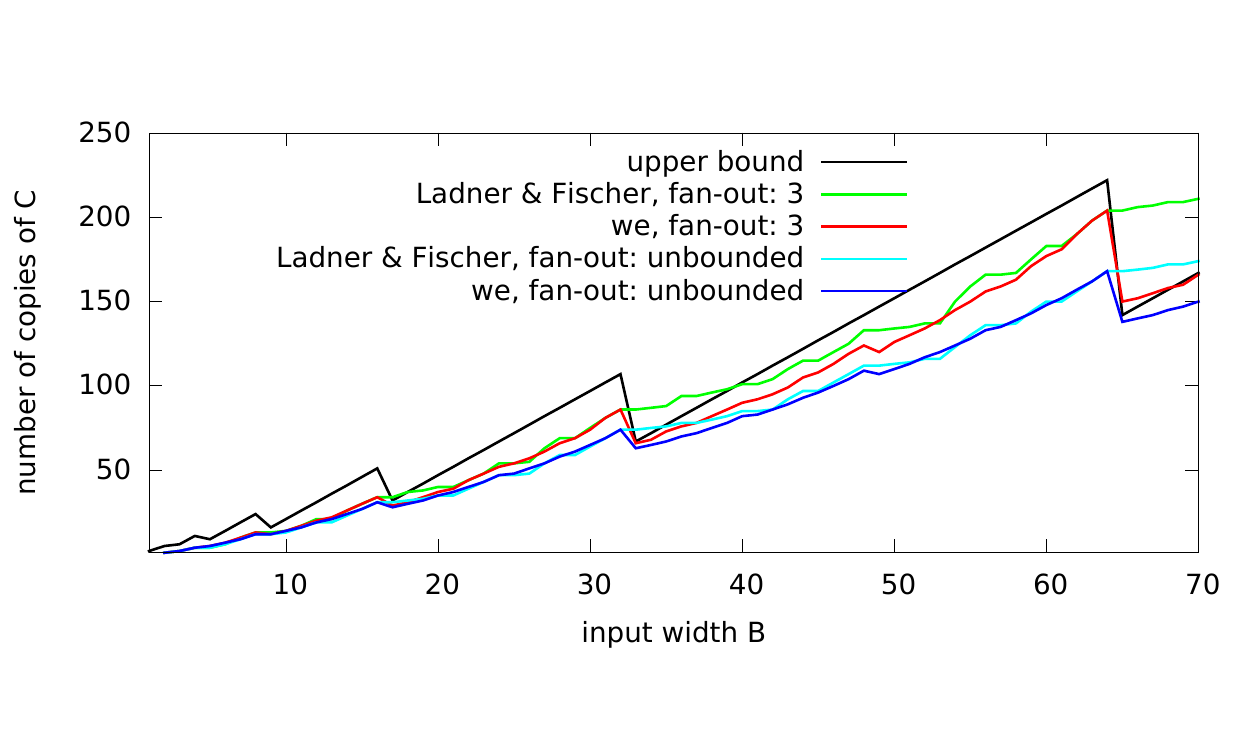}
\caption{Comparison of the balanced recursion from~\cite{ladner1980parallel} and
ours. The curves for unbounded fan-out are the exact sizes obtained, whereas
``upper bound'' refers to the bound from Corollary~\ref{cor:rdepth}; the fan-out
$3$ curves show that the unbalanced strategy performs better also for the
construction from Theorem~\ref{thm:ppc_fanout} (for $f=3$ and $k=0$) we derive
next.}
\label{fig:improv}
\end{figure}

\subsection{Constant Fan-out at Optimal Depth}\label{sec:fanout}

\begin{figure*}
\centering
\input{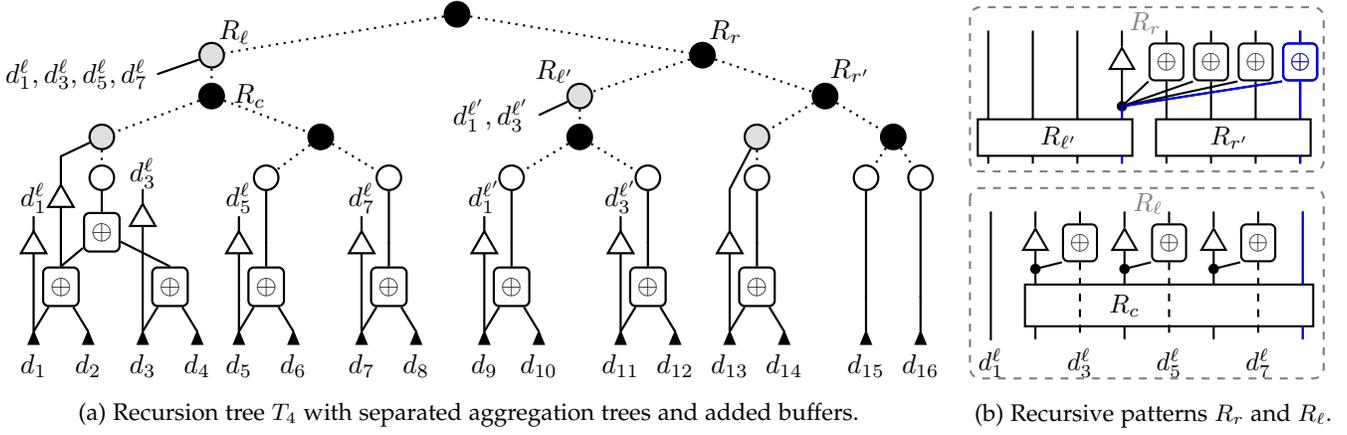}
\caption{Construction of $\ppc(C,T_4)'$. On the left, we see the recursion tree,
with the aggregation trees separated and shown at the bottom. Inputs are
depicted as black triangles. On the right, the application of the recursive
patterns at the children of the root is shown. Parts marked blue will be
duplicated in the second step of the construction that achieves constant
fan-out; this will also necessitate to duplicate some gates in the aggregation
trees.}
\label{fig:rectree}
\end{figure*}

The optimal depth construction incurs an excessively large fan-out of
$\Theta(B)$, as the last output of left recursive calls needs to drive all the
copies of $C$ that combine it with each of the corresponding right call's
outputs. This entails that, despite its lower depth, it will not result in
circuits of smaller physical delay than simply recursively applying the
construction from Figure~\ref{subfig:pattern1}. Naturally, one can insert buffer
trees to ensure a constant fan-out (and thus constantly bounded ratio between
delay and depth), but this increases the depth to $\Theta(\log^2 B + d(C)\log
B)$.

We now modify the recursive construction to ensure a constant fan-out, at the
expense of a limited increase in size of the circuit. The result is the first
construction that has size $\BO(B)$, optimal depth, and constant fan-out.

In the following, we denote by $f\geq 3$ the maximum fan-out we are trying to
achieve, where we assume that gates or memory cells providing the input to the
circuit do not need to drive any other components. For simplicity, we consider
$C$ to be a single gate, i.e., a gate driving two $C$ components has exactly
fan-out $2$.

We proceed in two steps. First, we insert $2B$ buffers into the circuit,
ensuring that the fan-out is bounded by $2$ everywhere except at the gate
providing the last output of each subcircuit corresponding to a left node.
In the second step, we will resolve this by duplicating these gates sufficiently
often, recursively propagating the changes down the tree. Neither of these
changes will affect the output (i.e.\ the correctness) of the circuit or its
depth, so the main challenges are to show our claim on the fan-out and bounding
the size of the final circuit.

\subsubsection{Step~1: Almost Bounding Fan-out by $2$}

Before proceeding to the construction in detail, we need some structural
insight on the circuit.
\begin{definition}\label{def:range}
For node $v\in T_b$, define its \emph{range $R_v$} and \emph{left-count
$\alpha_v$} recursively as follows.
\begin{itemize}
  \item If $v$ is the root, then $R_v=[1,2^b]$ and $\alpha_v=0$.
  \item If $v$ is the left child of $p$ with $R_p=[i,i+j]$, then
  $R_v=[i,i+(j+1)/2]$ and $\alpha_v=\alpha_p$.
  \item If $v$ is the right child of right node $p$ with $R_p=[i,i+j]$, then
  $R_v=[i+(j+1)/2+1,i+j]$ and $\alpha_v=\alpha_p$.
  \item If $v$ is the right child of left node $p$, then $R_v=R_p$ and
  $\alpha_v=\alpha_p+1$.
\end{itemize}
\end{definition}
Hence, the left-count $\alpha_v$ tells us for every node $v\in T_b$ the number
of left recursion steps preceding $v$, whereas $R_v$ gives us information
about the range of inputs used at node $v$. We observe that each recursion halves the
number of inputs and that the range is only cut in half if $\alpha_v$ does not
increase. Combining these observations with structural insights on the recursion
patterns in Figures~\ref{subfig:pattern1} and~\ref{subfig:pattern2}, we state
the following four properties of $\ppc(C,T_b)$.
\begin{lemma}\label{lem:structure}
Suppose the subcircuit of $\ppc(C,T_b)$ represented by node $v\in
T_b$ in depth $d\in [b+1]$ has range $R_v=[i,i+j]$. Then
\begin{enumerate}[(i)]
  \item it has $2^{b-d}$ inputs,
  \item $j=2^{b-d+\alpha_v}-1$,
  \item if $v$ is a right node, all its inputs are outputs of its childrens'
  subcircuits, and
  \item if $v$ is a left node or leaf, only its even inputs are provided by
  its child (if it has one) and for odd $k\in [1,2^{b-d}]$, we have that
  $d_k^v=\bigoplus_{k'=i+(k-1)2^{\alpha_v}}^{i+k2^{\alpha_v}-1}d_{k'}$.
\end{enumerate}
\end{lemma}
\begin{proof}
  Property (i) is immediate from the fact that with each step of recursion, the
  number of input and output wires is cut in half. Checking the above definition,
  we see that the range stays the same whenever $\alpha_v$ increases, and otherwise
  is halved on each recursive step; Property (ii) follows. Property (iii) can be
  readily verified from Figure~\ref{subfig:pattern2}.

  The final property is shown by induction on $b$. It is immediate for $b=0$ and
  $b=1$. For $b\geq 2$, the subcircuit of the left child $\ell$ of the root has
  $2^{b-1}$ inputs, the odd ones of which are inputs to the overall circuit
  (cf.~Figures~\ref{subfig:pattern1} and~\ref{subfig:pattern2}). As we have
  $\alpha_v=0$, we get that
  $d_k^{\ell}=d_{i+k}=\bigoplus_{k'=i+(k-1)2^{\alpha_v}}^{i+k2^{\alpha_v}-1}d_{k'}$
  and the node satisfies the claim. For the subcircuit $R_r$ corresponding to the
  subtree rooted at the right child of the root, the claim holds by the induction
  hypothesis applied to $b-1$. For the subcircuit $R_{\ell}$ of the left child, we
  see from Figure~\ref{subfig:pattern1} that the subcircuit $R_c$ corresponding to the
  subtree rooted at its child $c$ receives inputs $d_k^c=d_{2k-1}\oplus d_{2k}$,
  $k\in [1,2^{b-1}]$. Combining this with the induction hypothesis for $b-2$, the
  induction step is completed also in this case.
\end{proof}

Lemma~\ref{lem:structure} leads to an alternative representation of the circuit
$\ppc(C,T_b)$, see Figure~\ref{fig:rectree}, in which we separate gates in the
recursive pattern from Figure~\ref{subfig:pattern1} that occur before the
subcircuit $R_c$. Adding the buffers we need in our construction, this results
in the modified patterns given in Figure~\ref{subfig:recpatterns}. The separated
gates appear at the bottom of Figure~\ref{subfig:rectree2}: for each leaf $v$ of
$T_b$, there is a tree of depth $\alpha_v$ aggregating all of the circuit's
inputs from its range. Each non-root node in an aggregation tree provides its
output to its parent. In addition, one of the two children of an inner node in
the tree must provide its output as an input to one of the subcircuits
corresponding to a node of $T_b$, cf.~Property~(iv) of Lemma~\ref{lem:structure}.

From this representation, we will derive that the following modifications of
$\ppc(C,T_b)$ result in a $\ppc_{\oplus}(2^b)$ circuit $\ppc(C,T_b)'$, for which
a fan-out larger than $2$ exclusively occurs on the last outputs of subcircuits
corresponding to nodes of $T_b$.
\begin{compactenum}
  \item Add a buffer on each wire connecting a non-root node of any of the
  aggregation trees to its corresponding subcircuit (see
  Figure~\ref{subfig:rectree2}).
  \item For the subcircuit corresponding to left node $\ell$ with range
  $R_{\ell}=[i,i+j]$, add for each even $k\leq j$ (i.e., each even $k$ but the
  maximum of $j+1$) a buffer before output $\pi^{\ell}_k$ (see bottom of
  Figure~\ref{subfig:recpatterns}).
  \item For each right node $r$ with range $[i,i+j]$, add a buffer before output
  $\pi^r_{(j+1)/2}$ (see top of Figure~\ref{subfig:recpatterns}).
\end{compactenum}
\begin{lemma}\label{lem:fanout_almost}
With the exception of gates providing the last output of subcircuits
corresponding to nodes of $T_b$ (blue in Figure~\ref{subfig:recpatterns}),
fan-out of $\ppc(C,T_b)'$ is $2$. Buffers or gates driving an output of the
circuit drive nothing else.
\end{lemma}
\begin{proof}
First, we prove the following invariant: If each input to a subcircuit of
$\ppc(C,T_b)'$ corresponding to a node of $T_b$ is driven by a gate or
buffer driving no other wires, the same holds true for the outputs of
the subcircuit.
Suppose for contradiction that this invariant is violated and
  consider a minimal subcircuit doing so. There are three cases.
  \begin{itemize}
    \item The subcircuit corresponds to a leaf. This is a contradiction, as then
    the subcircuit simply is a wire connecting its sole input to its output.
    \item The subcircuit corresponds to a right node $r$
    (cf.~top of Figure~\ref{subfig:recpatterns}) with range $[i,i+j]$. As the invariant
    applies to the subcircuit corresponding to its left child, its outputs
    $\pi^r_1,\ldots,\pi^r_{(j-1)/2}$ satisfy the invariant. Its output
    $\pi^r_{(j+1)/2}$ satisfies the invariant due to the inserted buffer. As the
    remaining outputs are driven by gates that drive nothing else, this case also
    leads to a contradiction.
    \item The subcircuit corresponds to a left node $\ell$
    (cf.~bottom of Figure~\ref{subfig:recpatterns}) with range $[i,i+j]$. As
    $d^{\ell}_1$ is simply wired to $\pi^{\ell}_1$ (and nothing else), it
    satisfies the invariant. The last output $\pi^{\ell}_{j+1}$ satisfies the
    invariant, because the recursively used subcircuit does. The remaining outputs
    are driven by gates or buffers driving nothing else, again resulting in a
    contradiction.
  \end{itemize}
As all cases result in a contradiction, the invariant holds.

Next, observe that, by construction, the aggregation trees have fan-out $2$
after buffer insertion. Each buffer or gate from this part of the circuit
drives exactly one wire connecting it to the remainder of the circuit. Thus, the
above invariant shows that all subcircuits corresponding to nodes of $T_b$
satisfy that each of their outputs is driven by gate or a buffer driving nothing
else. Checking Figure~\ref{subfig:recpatterns}, we can thus conclude that indeed
no gate or buffer drives more than two others, unless it provides the last
output to one of the recursively used subcircuits in the construction at a right
node; gates or buffers driving an output of $\ppc(C,T_b)'$ drive only this
output.
\end{proof}
It remains to count the inserted buffers. We do so by computing a closed
form expression from the linear recurrence that describes the number of nodes of
a given type (left, right, leaf) in a given depth as function of the previous
one.
The following helper statement will be
useful for this, but also later on.
\begin{lemma}\label{lem:leaves}
Denote by $L_b\subseteq T_b$ the set of leaves of $T_b$. Then
$|L_b|=F_{b+2}$ and $\sum_{v\in L_b}2^{\alpha_v}=2^b$.
\end{lemma}
\begin{proof}
We have that $|L_0|=1=F_2$, $|L_1|=2=F_3$, and, by
Observation~\ref{def:fibonacci}, for $b\geq 2$ that $|L_b|=|L_{b-1}|+|L_{b-2}|$.
This recurrence has solution $|L_b|=F_{b+2}$.

Next, consider the recurrence given by $L_0'=1$, $L_1'=2$, and
$L_b'=L_{b-1}'+2L_{b-2}'$ for $b\geq 2$; the factor of $2$ assigns twice the
weight to the subtree rooted at the child of the root's left child, thereby
ensuring that each leaf is accounted with weight $2^{\alpha_v}$. This
recurrence has solution $2^b$.
\end{proof}

\begin{lemma}\label{lem:buffers}
Denote by $s$ the size of a buffer. Then
\begin{equation*}
|\ppc(C,T_b)'|= |\ppc(C,T_b)|+\left(2^b+2^{b-1}-F_{b+3}\right)s\,.
\end{equation*}
\end{lemma}

\begin{proof}
  To count the number $c(b)$ of buffers in subcircuits corresponding to nodes of
  $T_b$, we observe that $c(0)=0$, $c(1)=1$, and for $b\geq 2$ it holds that
  $c(b)=2^{b-2}+c(b-1)+c(b-2)$: $1$ buffer at the root (see top of
  Figure~\ref{subfig:recpatterns}), $2^{b-2}-1$ buffers at its left child (see
  bottom of Figure~\ref{subfig:recpatterns}), $c(b-1)$ buffers for the subtree rooted at
  its right child, and $c(b-2)$ buffers for the subtree rooted at the child of its left
  child. This recurrence relation has the solution $c(b)=2^b-F_{b+1}$.

  To count the number of buffers attached to the aggregation trees, recall that
  from depth $d\neq 0$ of each tree, exactly half of the nodes' output is required
  by a buffer connected to some node in $T_b$ (this follows from
  Lemma~\ref{lem:structure} and the fact that ranges of nodes in the same depth of
  $T_b$ form a partition of $[1,2^b]$). Thus, this number equals $\sum_{v\in
  L_b}2^{\alpha_v-1}-1$. By Lemma~\ref{lem:leaves}, the total number of buffers thus
  equals
  \begin{equation*}
  2^b-F_{b+1}+2^{b-1}-F_{b+2}=2^b+2^{b-1}-F_{b+3}\,.\qedhere
  \end{equation*}
\end{proof}
Similar arguments serve later as well.
The main reason why we will define
the function $a(v)$ in the next section without rounding is to ensure that we
again obtain linear recurrences, which can be solved using standard techniques
from linear algebra. As a downside, this results in slightly overestimating the
size of circuits, as we may ask for more copies of gates from children than are
actually needed.

\subsubsection{Step~2: Bounding Fan-out by $f$}

\begin{figure*}
\begin{center}
\input{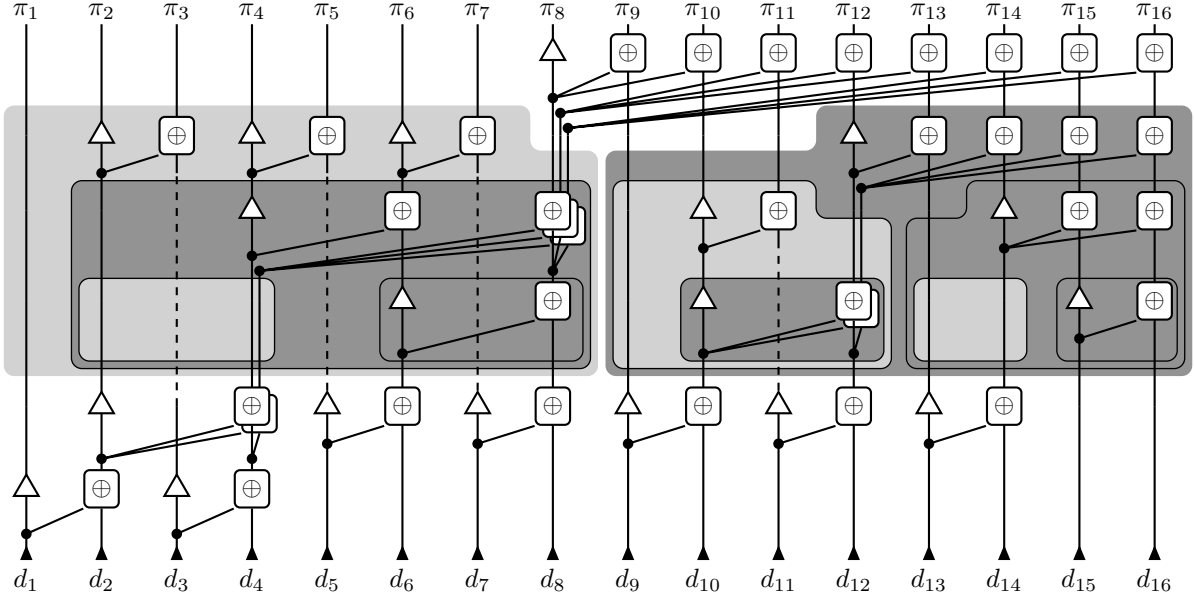}
\end{center}
\caption{$\ppc^{(3)}(C,T_4)$. Right recursion steps $R_r$ are
marked with dark gray, left recursion steps with light gray. The
step at the root (above) and aggregation trees (below) are not marked
explicitly. Duplicated gates are depicted in a layered fashion.
Dashed lines indicate that a wire is not participating in a recursive step.}
\label{fig:treeex}
\end{figure*}

In the second step, we need to resolve the issue of high fan-out of the last
output of each recursively used subcircuit in $\ppc(C,T_b)'$. Our approach is
straightforward. Starting at the root of $T_b$ and progressing downwards, we
label each node $v$ with a value $a(v)$ that specifies a sufficient number of
additional copies of the last output of the subcircuit represented by $v$ to
avoid fan-out larger than $f$. At right nodes, this is achieved by duplicating
the gate computing this output sufficiently often, marked blue in
Figure~\ref{subfig:recpatterns} (top). For left nodes, we simply require the
same number of duplicates to be provided by the subcircuit represented by their
child (i.e., we duplicate the blue wire in the bottom recursive pattern shown in
Figure~\ref{subfig:recpatterns}). Finally, for leaves, we will require a
sufficient number of duplicates of the root of their aggregation tree; this, in
turn, may require to make duplicates of their descendants in the aggregation
tree.

We define $a(v)$ and then utilize it to describe our fan-out $f$ circuit.
Afterwards, we will analyze the increase in size of the circuit compared to
$\ppc(C,T_b)'$.

\begin{definition}[$a(v)$]
Fix $b\in \NN_0$. For $v\in T_b$ in depth $d\in [b+1]$, define
\begin{equation*}
a(v)\coloneqq \begin{cases}
0 & \mbox{if }v\mbox{ is the root}\\
\frac{a(p)+2^{b-d}}{f} & \mbox{if }v\mbox{ is the left child of }p\\
\frac{a(p)}{f} & \mbox{if }v\mbox{ is the right child of right node }p\\
a(p) & \mbox{if }v\mbox{ is the (only) child of left node }p.\\
\end{cases}
\end{equation*}
\end{definition}
\begin{lemma}\label{lem:fanout_Tb}
Suppose that for each leaf $v\in T_b$, there are $\lfloor a(v)\rfloor$
additional copies of the root of the aggregation tree, and for each right node
$v\in T_b$, we add $\lfloor a(v)\rfloor$ gates that compute (copies of) the last
output of their corresponding subcircuit of $\ppc(C,T_b)'$. Then we can wire the
circuit such that all gates that are not in aggregation trees have fan-out at
most $f$, and each output of the circuit is driven by a gate or buffer driving
only this output.
\end{lemma}

\begin{proof}
We prove the claim by induction on the depth of nodes, starting at the leaves.
The base case holds by assumption, as each leaf is provided with sufficiently
many copies of its (single) input. For the induction step from depth $d+1$ to
$d\in [b]$, fix some $v\in T_b$ in depth $d$. By Lemma~\ref{lem:fanout_almost},
we only need to consider the last output of the subcircuit corresponding to the
node; there are in total $1+\lfloor a(v)\rfloor$ gates providing it. We
distinguish four cases.
\begin{itemize}
  \item $v$ is a left node. Thus, its child $c$ is a right node with
  $1+\lfloor a(c)\rfloor = 1+\lfloor a(v)\rfloor$ gates providing its last
  output. As the parent $p$ of $v$ is a right node, these need to drive
  $1+2^{b-d}+\lfloor a(p)\rfloor$ gates (cf.~Figure~\ref{subfig:recpatterns}).
  We have that
  \begin{align*}
  f(1+\lfloor a(v)\rfloor)&=f\left(1+\left\lfloor
  \frac{a(p)+2^{b-d}}{f}\right\rfloor\right)\\
  &\geq f+a(p)+2^{b-d}-(f-1)\\
  &\geq 1+2^{b-d}+\lfloor a(p)\rfloor\,.
  \end{align*}
  \item $v$ is a right node with left parent $p$. This was already
  covered in the previous case from the viewpoint of~$p$.
  \item $v$ is a right node with right parent $p$. As $p$ is also a right
  node, we need to drive $1+\lfloor a(p)\rfloor$ gates
  (cf.~Figure~\ref{subfig:recpatterns}). We have that
  \begin{align*}
  f(1+\lfloor a(v)\rfloor)&=f\left(1+\left\lfloor
  \frac{a(p)}{f}\right\rfloor\right)\\
  &\geq f+a(p)-(f-1) \geq 1+\lfloor a(p)\rfloor\,.
  \end{align*}
  \item $v$ is the root. Thus, it provides the outputs to the circuit, and the
  claim is immediate from Lemma~\ref{lem:fanout_almost}.\qedhere
\end{itemize}
\end{proof}

It remains to modify the aggregation trees so that sufficiently many copies of
the roots' output values are available.
\begin{lemma}\label{lem:aggregation}
Consider an aggregation tree corresponding to leaf $v\in T_b$ and fix $f\geq 3$.
We can modify it such that the fan-out of all its non-root nodes becomes at most
$f$, there are $\lfloor a(v)\rfloor$ additional gates computing the
same output as the root, and at most $(f a(v))/(f-2)+(2^{\alpha_v-1})/(f-1)$
gates are added.
\end{lemma}
\begin{proof}
  We recursively assign a value $a_d$ to each depth $d\in [\alpha_v+1]$ of the
  aggregation tree that bounds the number of additional gates in this
  depth from above. Accordingly, $a_0\coloneqq a(v)$. To determine a suitable
  recurrence, recall that for each node in the tree, one child needs to also drive
  a buffer, while the other does not (Lemma~\ref{lem:structure}). Fix one child
  $c$ in depth $d$ and its parent $p$, and denote by $a(c)$ and $a(p)$ the number
  of additional gates required. If $c$ also needs to drive a buffer, it is
  sufficient that $a(c)\geq \lfloor (a(p)+1)/f\rfloor$, as
  $f(1+\lfloor (a(p)+1)/f\rfloor)\geq 2+a(p)$;
  similarly, $a(c)\geq \lfloor a(p)/f\rfloor$ is sufficient, if the child does not
  need to drive an additional buffer. As aggregation trees are complete balanced
  binary trees, summing over all nodes in depth $d$ we thus get that
  $a_{d+1}=(2a_d+2^d)/f$ for all $d\in [\alpha_v]$ is sufficient. This recurrence
  has solution
  \begin{equation*}
  a_d = \left(\frac{2}{f}\right)^d a(v)
  +\frac{2^{d-1}}{f-1}\left(1-\frac{1}{f^d}\right)\,.
  \end{equation*}
  The total number of additional gates up to depth $\alpha_v-1$ is thus
  bounded by
  \begin{align*}
  \sum_{d=0}^{\alpha_v-1}a_d &=\sum_{d=0}^{\alpha_v-1}\left(\frac{2}{f}\right)^d a(v)
  +\frac{2^{d-1}}{f-1}\left(1-\frac{1}{f^d}\right)\\
  &<\frac{f a(v)}{f-2}+\frac{2^{\alpha_v-1}}{f-1}\,,
  \end{align*}
  matching the claimed bound.

  It remains to show that no gates need to be provided at the leaves of
  the aggregation tree; beside showing the claimed bound on the increase in size
  of the circuit, this is also necessary, because we cannot make copies of inputs
  without increasing the depth of the circuit. As $f\geq 3$ and nodes in
  aggregation trees have fan-out $2$ in $\ppc(C,T_b)'$, we can at least double the
  number of copies of each node with each level of the aggregation tree. As the
  tree at leaf $v$ has depth $\alpha_v$, it is hence sufficient to show that
  $a(v)\leq 2^{\alpha_v}$.

  To bound $a(v)$ from above, we again exploit the recursive structure of the
  construction. Denote by $A(b,x)$ an upper bound on $a(v)/2^{\alpha_v}$ for all
  leaves $v\in T_b$ when defining the values $a(v)$ as usual, except that we set
  $a(r)=x$ for some $0\leq x\leq 2^{b-1}$ at the root $r\in T_b$. For $b\geq 2$,
  we get that
  \begin{align*}
  A(b,x)&\leq \max\left\{
  A\left(b-1,\frac{x}{f}\right),A\left(b-2,\frac{2^{b-1}+x}{2f}\right)
  \right\}\\
  &< \max\left\{
  A\left(b-1,2^{b-2}\right),A\left(b-2,2^{b-2}\right)\right\}.
  \end{align*}
  Moreover, we have that $A(0,x)=x\leq 1/2$ and $A(1,x)=(x+1)/f\leq 2/3$ for
  feasible values of $x$, where we use that $f\geq 3$. Hence, $A(b,x)\leq 2/3$
  for all $b$ and feasible values of $x$. In particular, $A(b,0)\leq 2/3$,
  implying that indeed $a(v)<2^{\alpha_v}$ for all leaves $v\in T_b$.
\end{proof}

Finally, we need to count the total number of gates we add when implementing
these modifications to the circuit.
\begin{lemma}\label{lem:Delta}
For $f\geq 3$, define $\ppc^{(f)}(C,T_b)$ by modifying $\ppc(C,T_b)'$ according
to Lemmas~\ref{lem:fanout_Tb} and~\ref{lem:aggregation}. Then, with
$\lambda_1\coloneqq (1+\sqrt{5})/4$, $|\ppc^{(f)}(C,T_b)|$ is bounded by
\begin{equation*}
|\ppc(C,T_b)'|+2^b\left(\frac{1}{2f-2}+\frac{2}{f-2}+\BO\left(\frac{\lambda_1^b}{f^2}\right)\right)|C|\,.
\end{equation*}
\end{lemma}

\begin{proof}
Denote by $R_b\subset T_b$ the set of right nodes in $T_b$. By
Lemma~\ref{lem:fanout_Tb}, we add at most $\sum_{v\in R_b} a(v)$ gates
to the part of the circuit corresponding to $T_b$. Recall that $L_b\subseteq
T_b$ is the set of leaves of $T_b$. By Lemmas~\ref{lem:leaves}
and~\ref{lem:aggregation}, we add at most
\begin{equation*}
\sum_{v\in L_b}
\frac{fa(v)}{f-2}+\frac{2^{\alpha_v-1}}{f-1}=
\frac{2^{b-1}}{f-1}+\sum_{v\in L_b}\frac{fa(v)}{f-2}
\end{equation*}
gates to the aggregation trees.

To bound these numbers, denote for
$d\in [b]$ by $x_d$ and $y_d$ the sum over right and left nodes' $a(v)$ values
in depth $d$, respectively. Moreover, let $l_d$ and $r_d$ denote $2^{b-d}$
times the number of left and right nodes in depth $d$, respectively. Thus, we
seek to bound $f(x_b+y_b)/(f-2)+\sum_{d=0}^{b-1}x_d$. We have that $x_0=y_0=0$,
$l_1=r_1=2^{b-1}$, and
\begin{align*}
\begin{pmatrix}
x_{d}\\
y_{d}\\
l_{d+1}\\
r_{d+1}
\end{pmatrix}
&=
\begin{pmatrix}
f^{-1}& 1 & 0 & 0 & \\
f^{-1}& 0 & f^{-1} & 0\\
0 & 0 & 0 & 1/2\\
0 & 0 & 1/2 & 1/2
\end{pmatrix}^d
\begin{pmatrix}
x_0\\
y_0\\
l_1\\
r_1
\end{pmatrix}
\end{align*}
for all $d\in [b-1]$. The recurrence matrix has the eigenvalues
\begin{align*}
\lambda_1 = \frac{1}{4}(1+\sqrt{5})\,,\quad
\lambda_3=\frac{1}{2f}(1+\sqrt{1+4f})\,,\\
\lambda_2 = \frac{1}{4}(1-\sqrt{5})\,,\quad
\lambda_4=\frac{1}{2f}(1-\sqrt{1+4f})\,.
\end{align*}
The recurrence has solution $r_d=2^{b-d}F_{d+1}$, $l_d=2^{b-d}F_d$,
\begin{align*}
x_d =\,&
\frac{2^{b+1}}{f^2-10 f+20}\cdot
\bigg(
\frac{f^2-3 f-2}{\sqrt{1+4f}}\left(-\lambda_3^d+\lambda_4^d\right)-\\
&(f-2)\left(\lambda_1^d+\lambda_2^d-\lambda_3^d-\lambda_4^d\right)+\\
&\frac{3 f-10}{\sqrt{5}}\left(\lambda_1^d-\lambda_2^d\right)
\bigg)
\in\BO\left(\frac{2^b\lambda_1^d}{f}\right),
\end{align*}
and $y_d=x_{d+1}-x_d/f$, where the asymptotic bound on $x_d$ uses that for
$f\geq 3$ and all $i\in [1,4]$, $\lambda_1\geq |\lambda_i|$. Therefore,
\begin{align*}
\frac{f(x_b+y_b)}{f-2}+\sum_{d=0}^{b-1}x_d
&\in \BO\left(\frac{x_b+x_{b+1}}{f}\right)+\sum_{d=0}^{\infty}x_d\\
&=\BO\left(\frac{\lambda_1^d x_{d+1}}{f^2}\right)+\sum_{d=0}^{\infty}x_d\,,
\end{align*}
Observe that for $f\geq 3$, $0\neq |\lambda_i|<1$ for all $i$; thus,
$\sum_{d=0}^{\infty}\lambda_i^d=1/(1-\lambda_i)$, yielding with some calculation
that $\sum_{d=0}^{\infty}x_d=2^{b+1}/(f-2)$. Summation of the individual terms
and multiplication by $|C|$ proves claim of the theorem.
\end{proof}

As an example for the overall resulting construction, we show
$\ppc^{(3)}(C,T_4)$ in Figure~\ref{fig:treeex}. We summarize our findings in the
following theorem.
\begin{theorem}\label{thm:ppc_fanout}
Suppose that $C$ implements $\OP$, buffers have size $s$ and depth at most
$d(C)$, and set $\lambda_1\coloneqq (1+\sqrt{5})/4$. Then for all $k\in [b+1]$,
$b\in \NN_0$, and $f\geq 3$, there is a $\ppc_{\OP}(2^b)$ circuit of fan-out
$f$, depth $(b+k)d(C)$, and size at most
\begin{align*}
&\left(2^{b+1}+2^{b-k}\left(2+\frac{5f-6}{2f^2-6f+4}
+\BO\left(\frac{\lambda_1^b}{f^2}\right)\right)\right)|C|\\
&+\left(2^b+2^{b-k-1}\right)s\,.
\end{align*}
\end{theorem}

\begin{proof}
We argue as for Theorem~\ref{thm:ppc}, but replace $\ppc(C,T_b)$ by
$\ppc^{(f)}(C,T_b)$, and need to make sure that we modify the first $k$ steps of
the recursion, where we apply the construction from
Figure~\ref{subfig:pattern1}, such that the fanout is at most $f$. In fact, we
will ensure a fanout of $2$ for this part of the construction. To this end, we
simply add a buffer before each output that is not directly driven by a copy of
$C$, as already indicated in the figure. This guarantees the invariant that all
inputs to and outputs of subcircuits are driven by an element not driving
anything else; for the $\ppc^{(f)}(C,T_b)$ subcircuit, this invariant holds by
Lemma~\ref{lem:fanout_almost}.

This adds in total $2^b-2^{b-k}$ buffers to the circuit: one for each output
wire minus one for each output wire of $\ppc^{(f)}(C,T_b)$. Thus, by
Theorem~\ref{thm:ppc}, the size of the circuit is bounded by
$\Delta+\left(2^b-2^{b-k}\right)s+
\left(2^{b+1}+2^{b-k+1}\right)|C|$,
where $\Delta\coloneqq |\ppc^{(f)}(C,T_b)|-|\ppc(C,T_b)|$. By
Lemmas~\ref{lem:buffers} and~\ref{lem:Delta}, $\Delta$ is bounded by
\begin{equation*}
2^{b-k}\left(\frac{1}{2f-2}+\frac{2}{f-2}
+\BO\left(\frac{\lambda_1^{b-k}}{f^2}\right)\right)|C|
+ \frac{3\cdot 2^{b-k}s}{2}\,.
\end{equation*}
Summation yields the claimed bound on the size of the circuit. The
depth bound and that we indeed get a $\ppc_{\OP}(2^b)$ circuit follow as in
Theorem~\ref{thm:ppc}, as the modifications to the construction affected neither
its depth nor its output.
\end{proof}

\begin{figure}
\centering
\includegraphics[width=\linewidth, clip, trim=4 29 9 35]{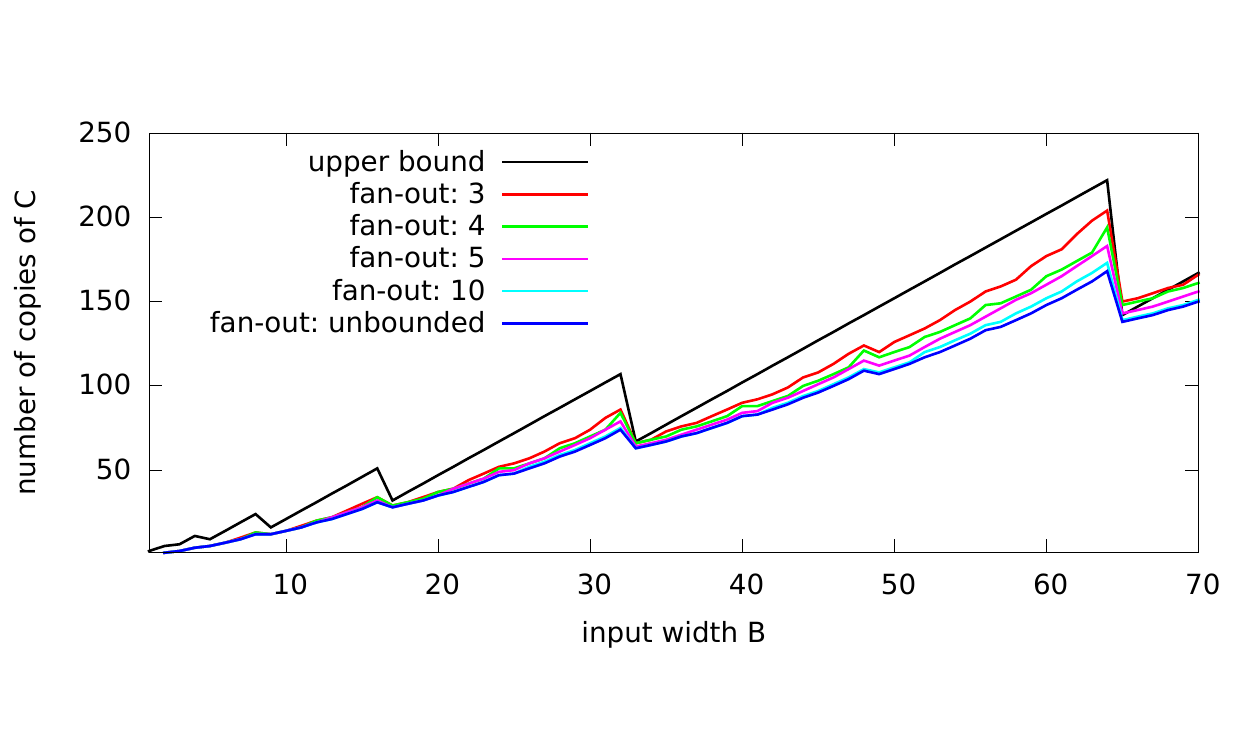}
\caption{Dependence of the size of the modified construction on $f$. For
comparison, the upper bound from Corollary~\ref{cor:rdepth} on the circuit with
unbounded fan-out is shown as well.}
\label{fig:fout}
\end{figure}
We refrain from analyzing the size of the construction for values of $B$ that
are not powers of $2$. However, in Figure~\ref{fig:fout} we plot the exact
bounds (without buffers) for $k=0$ and selected values of $f$ against $B$.

\section{Simulation}\label{sec:simulation}

Separate from and in addition to the proofs from the previous sections, we verify
the correctness of our circuits by VHDL simulation.
To this end, we first need
to specify implementations of the subcircuits computing $\diamond_{\metas}$ and
$\outt_{\metas}$.

\subsection{Gate-Level Implementation of Operators}\label{sec:impop}

From Tables~\ref{tab:diamond} and~\ref{tab:outt}, for $s,b\in \BB{2}$ we can
extract the Boolean formulas
\begin{align*}
(s\diamond b)_1 &= s_1\bar{s}_2 + s_1\bar{b}_1 + \bar{s}_2b_1\\
(s\diamond b)_2 &= \bar{s}_1s_2 + \bar{s}_1b_2 + s_2\bar{b}_2\\
\outt(s,b)_1 &= \bar{s}_1b_2 + \bar{s}_2b_1 + b_1b_2\\
\outt(s,b)_2 &= s_1b_2 + s_2b_1 + b_1b_2\,.
\end{align*}
In general, realizing a Boolean formula $f$ by replacing negation,
multiplication, and addition by inverters, $\ANDD$, and $\ORR$ gates,
respectively, does not result in a circuit implementing
$f_{\metas}$.\footnote{For instance, $(s\diamond b)_1 = s_1\bar{b}_1 +
\bar{s}_2b_1$ as Boolean formula, but the two expressions differ when evaluated
on $s_1=\bar{s}_2=1$ and $b_1=\metas$. The circuits resulting from the different
formulas are implementations of a multiplexer (with select bit $b_1$) and its
closure, respectively.} However, we can easily verify that the above formulas
are disjunctions of all prime implicants of their respective functions. As shown
in~\cite{friedrichs18containing} (see also
\cite{huffman57design}),\footnote{Alternatively, one can manually verify that
these formulas evaluate to the truth tables given in Tables~\ref{tab:diamondM}
and~\ref{tab:outtM}.} in this special case the resulting circuits do implement
the closure\dash---provided the gates behave as in Table~\ref{tab:gates}, which
the implementations given in Figure~\ref{fig:transistorgates} do by
Theorem~\ref{thm:gates_correct}. Using distributive laws (recall that these also
hold in Kleene logic), the above formulas can be rewritten as
\begin{align*}
(s\diamond b)_1 &= s_1(\bar{s}_2+\bar{b}_1) + \bar{s}_2b_1\\
(s\diamond b)_2 &= s_2(\bar{s}_1 + \bar{b}_2) + \bar{s}_1b_2\\
\outt(s,b)_1 &= b_1(b_2+\bar{s}_2) + b_2\bar{s}_1\\
\outt(s,b)_2 &= b_2(b_1+s_1) + b_1s_2\,.
\end{align*}
We see that, in fact, a single circuit with suitably wired (and possibly
negated) inputs can implement all four operations. As for
$\sel_1=\overline{\sel}_2$ the circuit implements a multiplexer with select bit
$\sel_1$, we refer to it as \emph{extended multiplexer,} or $\xmux$ for short.
Its functionality is specified by
\begin{equation*}
\xmux(\sel_1,\sel_2,x,y) \coloneqq y (x+\sel_2)+x \sel_1\,.
\end{equation*}
Figure~\ref{fig:xmux} shows the resulting circuit, and Table~\ref{tab:xmux}
lists how to map inputs to compute $\diamond_{\metas}$ and $\outt_{\metas}$.

We note that this circuit is not a particularly efficient $\xmux$
implementation; a transistor-level implementation would be much smaller.
However, our goal here is to verify correctness and give some initial indication
of the size of the resulting circuits\dash---a fully optimized ASIC circuit is
beyond the scope of this article. In~\cite{date18}, the size of the
implementation is slightly reduced by moving negations. Due to space
limitations, we refrain from detailing this modification here, but note that
Figure~\ref{fig:plot} and~Table~\ref{tab:sorting} take it into account.

\begin{figure}
\centering\small
\begin{tikzpicture}[circuit logic US, every circuit symbol/.style={thick}]
	\node (i1) at (0.7, 2) {$x$};
	\node (i0) at (2.2, 2) {$y$};
	\node (i2) at (0, -.1)  {$\sel_1$};
	\node (i3) at (0, 1.1) {$\sel_2$};
	\node (o0) at (6, .5) {};
	\node[or gate, inputs={nn}] (or1) at (1.5,1) {};
	\node[and gate, inputs={nn}] (and1) at (1.5,0) {};
	\node[and gate, inputs={nn}] (and2) at (3,1.1) {};
	\node[or gate, inputs={nn}] (or2)  at (4.5,.5) {};
        \draw[thick] (i3) -- (or1.input 1);
        \filldraw (.7, .9) circle (1pt);
        \draw[thick] (i1) ++(down:1.1) -| (or1.input 2);
        \draw[thick] (i1) -- ++(down:1.9) -| (and1.input 1);
        \draw[thick] (i2) -- (and1.input 2);
        \draw[thick] (or1.output) -- (and2.input 2);
        \draw[thick] (i0) -- ++(down:.8) -| (and2.input 1);
        \draw[thick] (and2.output) -- ++(right:.3) -- ++(down:.5) -| (or2.input 1);
        \draw[thick] (and1.output) -- ++(right:1.8) -- ++(up:.4) -| (or2.input 2);
        \draw[thick] (or2.output) -- (o0);
\end{tikzpicture}
\caption{$\xmux$ circuit, used to implement $\diamond_{\metas}$ and
$\outt_{\metas}$.}
\label{fig:xmux}
\end{figure}

\begin{table}
\centering
\caption{Wiring an $\xmux$ to compute the various operators.}\label{tab:xmux}
\begin{tabular}{|c |c |c |c ||c |}
	\hline
	$\sel_1$ & $\sel_2$ & $x$ & $y$ & $\xmux(\sel_1,\sel_2,a,b)$\\ \hline \hline
	\rule{0pt}{8pt} $b_1$ & $\bar{b}_1$ & $\bar{s}_2$ & $s_1$ & $(s \diamond_{\metas} b)_1$\\
	$b_2$ & $\bar{b}_2$ & $\bar{s}_1$ & $s_2$ & $(s \diamond_{\metas} b)_2$\\
	$\bar{s}_1$ & $\bar{s}_2$ & $b_2$ & $b_1$ & $\outt_{\metas}(s,b)_1$\\
	$s_2$ & $s_1$ & $b_1$ & $b_2$ & $\outt_{\metas}(s,b)_2$ \\
	\hline
\end{tabular}
\end{table}

\subsection{Putting it All Together}

We now have all the pieces in place to assemble a containing $\twosort(B)$
circuit. By Theorem~\ref{thm:assoc}, $\diamond_{\metas}$ is associative. Thus,
from a given implementation of $\diamond_{\metas}$ (e.g., two copies of
the circuit from Figure~\ref{fig:xmux} with appropriate wiring and negation,
cf.~Table~\ref{tab:xmux}) we can construct $\ppc_{\diamond_{\metas}}(B-1)$
circuits of small depth and size, as shown in Section~\ref{sec:ppc}. We can
combine such a circuit with an $\outt_{\metas}$ implementation (again, two
$\xmux$es with appropriate wiring and negation will do) as shown in
Figure~\ref{fig:sortppc} to obtain our $\twosort(B)$ circuit.

\subsection{Simulation Setup}

We implemented the design given in Figure~\ref{fig:sortppc} on
register-transfer-level using the $\ppc_{\diamond_{\metas}}(B-1)$ circuit given
by Theorem~\ref{thm:ppc} for $k=0$.\footnote{For $k>0$, fan-out becomes an
issue, requiring the more involved constructions provided by
Theorem~\ref{thm:ppc_fanout}. However, the resulting numbers would be
inaccurate, and a detailed comparison based on optimized ASIC implementations is
beyond the scope of this work.} \emph{Quartus} by {Altera} is used for design
entry, which in our case mainly consists of checking correct implementation.
After design entry we use \emph{ModelSim} by {Altera} for behavioral simulation.
Note that we must not simulate the preprocessed Quartus output, because
processing may compromise metastability-containing behavior. Instead, we
simulate pure VHDL. Metastable signals are simulated using VHDL signal $X$,
because its behavior matches the worst-case behavior assumed for $\metas$.

\begin{figure}\centering
\begin{tikzpicture}[scale=0.99]\small
\node (in1) at (0,5) {$g_1 h_1$};
\node (in2) at (2,5) {$g_{B-2}h_{B-2}$};
\node (in3) at (4,5) {$g_{B-1}h_{B-1}$};

\draw[thick] (-.5,2.4) rectangle (4.5,4);
\node at (2,3.2) {$\ppc_{\diamond_{\metas}}(B-1)$};
\node at (0,3.8) {\small$\inp_{1}$};
\node at (1,3.8) {\small$\hdots$};
\node at (2,3.8) {\small$\inp_{B-2}$};
\node at (4,3.8) {\small$\inp_{B-1}$};
\node at (0,2.6) {\small$\outp_{1}$};
\node at (1,2.6) {\small$\hdots$};
\node at (2,2.6) {\small$\outp_{B-2}$};
\node at (4,2.6) {\small$\outp_{B-1}$};

\node (zero) at (-2,2.5) {$00$};

\node (in5) at (-1,2) {$g_1h_1$};
\node (in6) at (1,2) {$g_2h_2$};
\node (in7) at (3,2) {$g_{B-1}h_{B-1}$};
\node (in8) at (5,2) {$g_{B}h_{B}$};

\node[rectangle, draw, rounded corners, thick] (sel1) at (-1,1) {$\outt_{\metas}$};
\node[rectangle, draw, rounded corners, thick] (sel2) at (1,1) {$\outt_{\metas}$};
\node at (2,1) {$\hdots$};
\node[rectangle, draw, rounded corners, thick] (sel3) at (3,1) {$\outt_{\metas}$};
\node[rectangle, draw, rounded corners, thick] (sel4) at (5,1) {$\outt_{\metas}$};

\node (out1) at (-1,0) {$g'_1h'_1$};
\node (out2) at (1,0) {$g'_2h'_2$};
\node (out3) at (3,0) {$g'_{B-1}h'_{B-1}$};
\node (out4) at (5,0) {$g'_{B}h'_{B}$};

\draw[thick] (in1) -- ++(down:1);
\draw[thick] (in2) -- ++(down:1);
\draw[thick] (in3) -- ++(down:1);
\draw[thick] (in5) -- (sel1);
\draw[thick] (in6) -- (sel2);
\draw[thick] (in7) -- (sel3);
\draw[thick] (in8) -- (sel4);
\draw[thick] (zero) -- (-2,1.8) -- (sel1);
\draw[thick] (0,2.4) -- (0,1.8) -- (sel2);
\draw[thick] (2,2.4) -- (2,1.8) -- (sel3);
\draw[thick] (4,2.4) -- (4,1.8) -- (sel4);
\draw[thick] (sel1) -- (out1);
\draw[thick] (sel2) -- (out2);
\draw[thick] (sel3) -- (out3);
\draw[thick] (sel4) -- (out4);
\end{tikzpicture}
\caption{Constructing $\twosort(B)$ from $\ppc_{\diamond_{\metas}}(B-1)$ and
$\outt_{\metas}$.}
\label{fig:sortppc}
\end{figure}

The correctness of this construction follows from Theorems~\ref{thm:statemetasi}
and~\ref{thm:out}, where we can plug in any $\ppc_{\diamond_{\metas}}(B-1)$
circuit, cf.~Section~\ref{sec:ppc}. For the circuits derived by relying on the
$\xmux$ circuit from Figure~\ref{fig:xmux}, we independently confirmed this
via simulation.

\subsection{Results}

\begin{figure}
\begin{center}
\includegraphics[trim={0cm 2cm 4.8cm 0cm}, clip, width=.8\columnwidth]{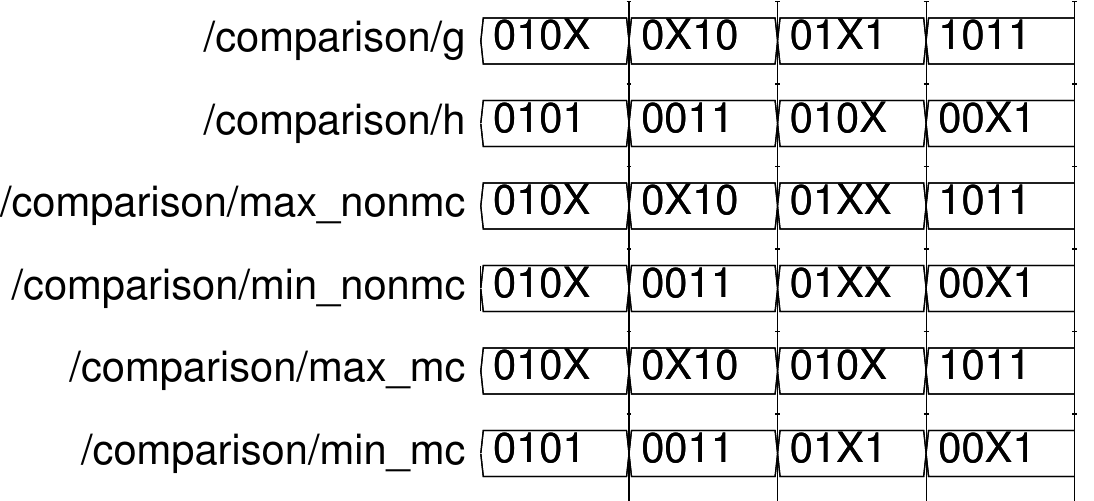}
\caption{Excerpt from a simulation for $4$-bit inputs, where $X=M$. The rows
show (from top  to bottom) the inputs $g$ and $h$, both outputs of the simple
non-containing circuit, and both outputs of our design. As inputs $g$ and $h$ we
randomly generated valid strings. Columns $1$ and $3$ show that the simpler
design fails to implement a $\twosort(4)$ circuit.}
\label{fig:wavereg}
\end{center}
\end{figure}

For the implementation of $\ppc_{\diamond_{\metas}}(B-1)$ we used the circuits
from Theorem~\ref{thm:ppc}, i.e., we did not make use of the extension to
constant fan-out. We exhaustively checked the design from
Figure~\ref{fig:sortppc} for $B$ up to $12$ (and all feasible $k$).
Simulation shows that the design works correctly for several levels of
recursion, e.g., when regarding $B=1$ and $B=2$ as simple base cases, $B=12$ implies $3$
levels of recursion for both patterns. We refrained from simulating the constant
fan-out construction, because it simply replicates intermediate results without
changing functionality.

\subsection{Comparison to Baseline}

After behavioral simulation, we continue with a comparison of our
design and a standard sorting approach $\binary(B)$. As mentioned earlier,
the $\twosort(B)$ implementation given in Figure~\ref{fig:sortppc} is
slightly optimized by pulling out a negation from the operators in every
recursive step~\cite{date18}.

After design entry as described above, we use \emph{Encounter
RTL Compiler} for synthesis and \emph{Encounter} for place and route. Both tools are part
of the Cadence tool set and in both steps we use NanGate $45$\,nm Open Cell Library as
a standard cell library.

Since metastability-containing circuits may include additional gates that are
not required in traditional Boolean logic, Boolean optimization may compromise
metastability-containing properties~\cite{date17}. Accordingly, we were forced
to disable optimization during synthesis of the circuits.

\paragraph*{\textbf{Binary Benchmark $\binary$}}

\begin{figure}
\centering
\includegraphics[width=\linewidth]{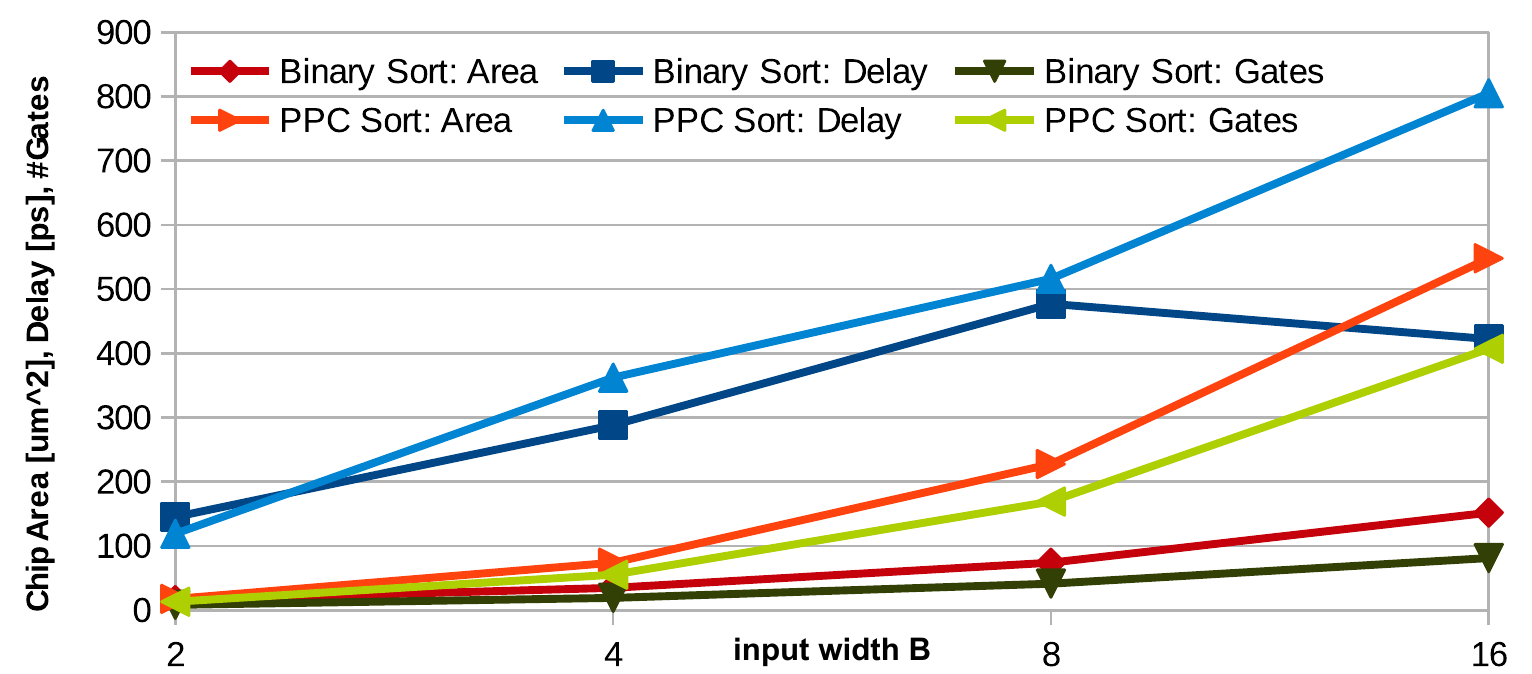}
\caption{Comparison of our solution PPC Sort to a standard non-containing one.
For the latter, the unexpected delay reduction at $B=16$ is the result of
automatic optimization with more powerful gates, which our solution does not
use.}
\label{fig:plot}
\end{figure}

In short, $\binary$ consists of a simple VHDL statement comparing two binary
encoded inputs and outputting the maximum and the minimum, accordingly. It
follows the same design process as $\twosort$, but then undergoes optimization
using a more powerful set of basic gates. For example, the standard cell library
provides prebuild multiplexers. These multiplexers are used by $\binary$, but
not by $\twosort$, as they are not metastability-containing. We stress that
these more powerful gates provide optimized
implementations of multiple Boolean functions, yet each of them is still counted
as a single gate. Thus, comparing our design to the binary design in terms of
gate count, area, and delay disfavors our solution. Moreover, we noticed that
the optimization routine switches to employing more powerful gates when going
from $B=8$ to $B=16$ (cf. Figure~\ref{fig:plot}), resulting in a \emph{decrease}
of the delay of the $\binary$ implementation.

Nonetheless, our design performs comparably to the non-containing binary design
in terms of delay, cf.~Figure~\ref{fig:plot} and Table~\ref{tab:sorting}. This
is quite notable, as further optimization is possible by optimizing our design
on the transistor level, with significant expected gains. The same applies to
gate count and area, where a notable gap remains. Recall, however, that the
$\binary$ design hides complexity by using more advanced gates and does not
contain metastability.

We emphasize that we refrained from optimizing the design by making use of all
available gates or devising transistor-level implementations, as
such an approach is tied to the utilized library or requires design of
standard cells.

\begin{table*}\scriptsize
\begin{center}
\caption{Simulation results for metastability-containing sorting networks with
$n\in \{4,7,10\}$ for $B$-bit inputs. $\tensortc$ optimizes gate count~\cite{codish2014twenty}, $\tensortd$ optimizes depth~\cite{bundala2014optimal}; for $n\in \{4,7\}$, the sorting networks are optimal w.r.t.\
both measures. Simulation results are: (i) number of gates, (ii) postlayout area $[\mu m^2]$ and
(iii) prelayout delay $[ps]$.}
\label{tab:sorting}
\begin{adjustbox}{center}
\begin{tabular}{| c | c ||  c | c | c || c| c |c || c| c | c || c | c | c |}
\hline
\multirow{2}{*}{$B$} & \multirow{2}{*}{Circuit}   & \multicolumn{3}{|c ||}{$\foursort$} & \multicolumn{3}{|c ||}{$\sevensort$} & \multicolumn{3}{|c ||}{$\tensortc$} & \multicolumn{3}{|c |}{$\tensortd$}\\
                 \cline{3-14}
     &           & gates & area & delay & gates & area & delay & gates & area & delay & gates & area & delay
\tabularnewline
\hline
\hline
\multirow{3}{*}{$2$}& our work & 65 & 87.402 & 357 & 208 & 279.741 & 714 & 377 & 506.912 & 912 & 403 & 541.968 & 833
\tabularnewline & $\binary$ & 40 & 77.91 &478 & 128 & 249.326 &953 & 232 & 451.815 &1284 & 248 & 483 &1145
\tabularnewline
\hline
\multirow{3}{*}{$4$}& our work & 275 & 368.641 & 640 & 880 & 1179.528 & 1014 & 1595 & 2137.905 & 1235 & 1705 & 2285.514 & 1133
\tabularnewline & $\binary$ & 95 & 172.935 & 906 & 304 & 553.28 & 1810 & 551 & 1002.848 & 2429& 589 & 1072.099 & 2143
\tabularnewline
\hline
\multirow{3}{*}{$8$}& our work & 845 & 1136.184 & 1396 & 2704 & 3636.08 & 1921 & 4901 & 6590.283 & 2179 & 5239 & 7044.541 & 2059
\tabularnewline & $\binary$ & 205 & 368.641 &1475 & 656 & 1179.528 &2948 & 1189 & 2137.905 &3945& 1271 & 2285.514 &3470
\tabularnewline
\hline
\multirow{3}{*}{$16$}& our work  & 2035 & 2739.961 & 2069  & 6512 & 8767.374 & 3396 & 11803 & 15891.12 & 4030 & 12617 & 16987.194 & 3844
\tabularnewline  & $\binary$ & 405 & 530.67 &1298 & 1296 & 2425.99 &2600 & 2349 & 4397.085 &3474 & 2511 & 4700.304 &3050
\tabularnewline
\hline
\end{tabular}
\end{adjustbox}
\end{center}
\end{table*}

\section{Conclusions}\label{sec:conclusion}
In this work, we demonstrated that efficient metastability-containing sorting
circuits are possible. Our results indicate that optimized implementations can
achieve the same delay as non-containing solutions, without a dramatic increase
in circuit size. This is of high interest to an intended application
motivating us to design MC sorting circuits: fault-tolerant high-frequency clock
synchronization. Sorting is a key step in envisioned implementations
(cf.~\cite{friedrichs18containing,huemer2016synchronization}) of the Lynch-Welch
algorithm~\cite{welch88} with improved precision of synchronization. The
complete elimination of synchronizer delay is possible due to the efficient MC
sorting networks presented in this article; enabling an increment of the rate at which
clock corrections are applied, significantly reducing the negative impact of
phase drift of local clock sources on the precision of the
algorithm~(cf.~\cite{khanchandani18}).

This goal will necessitate to devise optimized ASIC implementations of our
circuits. The novel PPC circuits we devised in Section~\ref{sec:ppc} are an
important contribution towards this end. Note that it is crucial to take
into account both depth and fan-out for devising low-delay circuits. Hence,
follow-up work needs to compare the existing and our novel design based on
suitable metrics that take both into account to reliably predict the achieved
trade-offs between delay, area, and energy consumption of circuits. Note that
this is of relevance beyond the specific application of MC sorting: PPC circuits
lie at the heart of adder designs, implying that even a minor improvement can
have significant impact on the overall performance of computing devices!

\paragraph*{\textbf{MC Control Loops}}
More generally speaking, MC circuits like those presented here are of interest
in mixed-signal control loops whose performance depends on very short response
times. When analog control is not desirable, traditional solutions incur
synchronizer delay before being able to react to any input change. Using MC
logic saves the time for synchronization, while metastability of the output
corresponds to the initial uncertainty of the measurement; thus, the same
quality of the computational result can be achieved in shorter time. Note that
our circuits are purely combinational, so they can be used in both clocked
and asynchronous control logic.

Obvious examples of such control loops are clock synchronization circuits, but
MC has been shown to be useful for adaptive voltage control~\cite{FKLW-async-18}
and fast routing with an acceptable low probability of data
corruption~\cite{TFL-async-17} as well. This type of application suggests to
explore whether efficient circuits exist for a wider range of arithmetic
operations, like e.g.\ addition or (possibly approximate) multiplication.

\paragraph*{\textbf{Redundant Encoding and Addition}}
On the theoretical side, our results are to be contrasted with the exponential
gap between the size of non-containing and MC circuits shown
in~\cite{ikenmeyer18complexity}. This work raised the question for which classes
of functions small MC circuits exist. Given that Ladner and Fischer proved that
the PPC task can be solved efficiently for any constant-sized state
machine~\cite{ladner1980parallel}, it was natural to ask whether this result can
be extended to MC computations. In follow-up work, we show that indeed this
holds true for any constant-sized FSM~\cite{BLM19}. However, when applying this
result to addition, unlike for sorting (where the underlying operations are
$\max$ and $\min$) uncertainty of inputs adds up. This means that Gray code can
support meaningful computations only if the \emph{total} uncertainty of all
addends is at most $1$.

Accordingly, in~\cite{BLM19} we also consider redundant encodings, showing that
using $k$ (roughly) redundant bits, an uncertainty of $\lfloor(k+1)/2\rfloor$
can be tolerated without loss of precision. Combined with the above result on
transducers, this yields a meaningful notion of MC addition that allows for
efficient circuits. As, essentially, the redundant bits are used as a unary
code, it should be straightforward to apply the techniques from this article
to obtain efficient sorting circuits with the encoding from~\cite{BLM19}. We
remark that the encoding from~\cite{BLM19} turns out to be identical to that of
the output of suitable time-to-digital converters~\cite{tdc16}, so relaxing
their output constraints to achieve better average-case performance would
provide valid input for sorting circuits that accept inputs encoded in this
manner.

We believe that these results suggest applicability of our techniques to a wide
range of mixed-signal control loops and call for future work further exploring
to which extend basic arithmetics can be realized by efficient MC circuits.

\paragraph*{\textbf{Acknowledgments}}
We thank Attila Kinali and the anonymous reviewers for valuable input. This
project has received funding from the European Research Council (ERC) under the
European Union's Horizon 2020 research and innovation programme (grant agreement
716562).

\bibliographystyle{plain}
\bibliography{comp_short}

\end{document}